\pdfoutput=1
\documentclass[11pt]{article}
\def\paperversion{tr}
\def\anonymoussubmission{1} 
\newif\ifsinglecolumn\singlecolumnfalse
\newif\ifwidemargins\widemarginsfalse
\newif\ifwarning\warningfalse
\newif\ifshowcomments\showcommentsfalse
\newif\ifblinded\blindedfalse
\newif\ifreport\reportfalse
\newif\ifcopyrightspace\copyrightspacefalse
\newif\ifacknowledgments\acknowledgmentsfalse
\newif\ifshowpagenumbers\showpagenumberstrue
\newif\iffinalformat\finalformatfalse
\newif\ifweb\webfalse

\ifx\anonymoussubmission\undefined
\else
  \if\anonymoussubmission 1
    \blindedtrue
  \else
    \blindedfalse
  \fi
\fi
\IfFileExists{.blinded}{\blindedtrue}

\ifx\paperversion\xxxxundefined
\PackageError{paperversions}{*** No valid document version was specified.
Macro paperversions must be be defined as one of (markup, draft,
local, submission, final, web, tr, trdraft)}
\fi

\def\xxversion{\csname xx\paperversion\endcsname}
\newif\ifsawversion\sawversionfalse

\ifcase\xxversion\relax
    \widemarginstrue
    \singlecolumntrue
    \warningtrue
    \sawversiontrue
\or 
    \warningtrue
    \showcommentstrue
    \sawversiontrue
\or 
    \warningtrue
    \showcommentsfalse
    \blindedfalse
    \sawversiontrue
    \acknowledgmentstrue
\or 
    \sawversiontrue
\or 
    \blindedfalse
    \sawversiontrue
    \copyrightspacetrue
    \acknowledgmentstrue
    \showpagenumbersfalse
    \finalformattrue
\or 
    \singlecolumntrue
    \blindedfalse
    \sawversiontrue
    \reporttrue
    \acknowledgmentstrue
    \webtrue
\or 
    \blindedfalse
    \singlecolumntrue
    \showcommentstrue
    \sawversiontrue
    \reporttrue
\or 
    \blindedfalse
    \showcommentstrue
    \copyrightspacetrue
    \acknowledgmentstrue
    \sawversiontrue
    \finalformattrue
\or 
    \blindedfalse
    \sawversiontrue
    \copyrightspacetrue
    \acknowledgmentstrue
    \finalformattrue
    \webtrue
\or 
    \singlecolumntrue
    \blindedtrue
    \acknowledgmentsfalse
    \sawversiontrue
    \reporttrue
    \webtrue
\fi

\ifsawversion
\else
\fi

\let\xxversion=\undefined

    \usepackage{times,fullpage}
    \usepackage[dvipsnames]{xcolor}
    \usepackage{amsthm}
 
\usepackage{times,graphicx,eso-pic,xspace}
\usepackage{amsmath, amssymb, stmaryrd, utf8math, tikz}
\usetikzlibrary{arrows} 
\usetikzlibrary{positioning}
\usetikzlibrary{decorations.text}
\usetikzlibrary{decorations.pathmorphing}
\usepackage{ifthen}
\usepackage[rightcaption]{sidecap}
\usepackage{multirow}

\usepackage{caption}
\DeclareCaptionFormat{lineAfter}{\tb{#1}#2#3\vspace{-7pt}\hrulefill}
\DeclareCaptionFormat{boring}{\tb{#1}#2#3}
\usepackage{subcaption}
\def\correctabovecaptionskip{1mm}
\def\correctbelowcaptionskip{-6mm}
\usepackage{hyperref}

\ifshowpagenumbers
  \def\thepage{\arabic{page}}
\fi

\ifshowcomments
    \usepackage{aplcomments}
\else
    \usepackage[disabled]{aplcomments}
\fi

\usepackage{listings}
\usepackage{color}
\definecolor{javared}{rgb}{0.6,0,0} 
\definecolor{javagreen}{rgb}{0.25,0.5,0.35} 
\definecolor{javapurple}{rgb}{0.5,0,0.35} 
\definecolor{javadocblue}{rgb}{0.25,0.35,0.75} 
\usepackage{inconsolata}

\usepackage{cleveref}
\crefname{section}{{\S\!}}{{{\S}{\S}}} 
\Crefname{section}{{Section}}{{Sections}}
\crefname{equation}{{Eqn.}}{{Eqns.}}
\Crefname{equation}{{Eqn.}}{{Eqns.}}
\crefname{theorem}{{Thm.}}{{Thms.}}
\Crefname{theorem}{{Thm.}}{{Thms.}}
\crefname{lemma}{{Lemma}}{{Lemmas}}
\Crefname{lemma}{{Lemma}}{{Lemmas}}
\crefname{figure}{{Fig.}}{{Figures}}
\Crefname{figure}{{Fig.}}{{Figures}}
\crefname{definition}{{Def.}}{{Definitions}}
\Crefname{definition}{{Def.}}{{Definitions}}

\def \correctbelowcaptionskip{\belowcaptionskip}
\ifreport
\title{Safe Serializable Secure Scheduling:\\
Transactions and the Trade-Off Between Security and Consistency\\
    {\Large (Technical Report)}
  }
\else
  \title{Safe Serializable Secure Scheduling:}
  \subtitle{Transactions and the Trade-Off Between Security and Consistency}
\fi

\ifblinded
    \numberofauthors{1}
    \author{
        Anonymized For Peer Review
     }
\else
  \ifreport
    \author{
        Isaac Sheff \qquad Tom Magrino \qquad Jed Liu \qquad
        Andrew C. Myers \qquad Robbert van Renesse \vspace{2mm}\\
        Department of Computer Science,\ \ \ 
        Cornell University, \ \ \ 
        Ithaca, New York, USA \\
        \normalsize\texttt{\{isheff,tmagrino,liujed,andru,rvr\}@cs.cornell.edu}
    }
  \else
    \numberofauthors{5}
    \author{
        Isaac Sheff\\
        \and
        Tom Magrino\\
        \and
        Jed Liu\\
        \and 
        Andrew C. Myers\\
        \and
        Robbert van Renesse\\
        \end{tabular} \\[6pt]
        \begin{tabular}{c}
        \affaddr{Cornell University ~~~ Department of Computer Science}
               \affaddr{~~~ Ithaca, New York, USA}\\
               \email{\normalsize\texttt{\{isheff,tmagrino,liujed,andru,rvr\}@cs.cornell.edu}}
    }
  \fi
\fi
\ifreport
\else
    \CopyrightYear{2016} 
    \setcopyright{acmcopyright}
    \conferenceinfo{CCS'16,}{October 24--28, 2016, Vienna, Austria}
    \isbn{978-1-4503-4139-4/16/10}\acmPrice{\$15.00}
    \doi{http://dx.doi.org/10.1145/2976749.2978415}
\fi

\ifwarning
\AtBeginDocument{
 \AddToShipoutPicture*{\put(306,690){\it\color{red}{\fbox{\large Draft---please do not distribute}}}}
}
\fi

\newcommand{\before}[1][\relax]{{\ensuremath{\ifx\relax#1\relax\else\overset{{#1}}\fi\rightarrowtriangle}}}
\newcommand{\p}[1]{{\left({{#1}}\right)}}
\newcommand{\cb}[1]{{\left\{{{#1}}\right\}}}
\newcommand{\sqb}[1]{{\left[{{#1}}\right]}}

\newcommand{\an}[1]{{\left\langle{{#1}}\right\rangle}}

\newcommand{\tb}[1]{{\textrm{\textbf{{#1}}}}}
\newcommand{\join}[0]{{\ensuremath\sqcup}}
\newcommand{\meet}[0]{{\ensuremath\sqcap}}
\newcommand{\less}[0]{{\ensuremath\sqsubseteq}}
\newcommand{\noinunbo}[1]{\textbf{{#1.}}}
\newtheorem{theorem}{Theorem}
\newtheorem{definition}{Def.}
\newtheorem{lemma}{Lemma}
\newcommand{\ti}[1]{{\emph{{#1}}}}
\newcommand{\relaxedsecurity}{relaxed observational determinism\xspace}
\newcommand{\Relaxedsecurity}{Relaxed observational determinism\xspace}
\newcommand{\RelaxedSecurity}{Relaxed Observational Determinism\xspace}
\newcommand{\ProtocolSecurity}[0]{Protocol Security\xspace} 
\newcommand{\ProtocolSecurityAdjective}[0]{secure\xspace} 
\newcommand{\TransactionSecurity}[0]{Secure Information Flow\xspace} 
\newcommand{\TransactionSecurityAdjective}[0]{in\-form\-ation-flow secure\xspace} 
\newcommand{\Patsy}[0]{Patsy\xspace}
\newcommand{\Attacker}[0]{Mallory\xspace}
\newcommand\NIE{\ensuremath\textit{NIE}}
\newcommand{\AttackerPossessivePronoun}{her} 
\newcommand{\AttackerPronounCapitalized}{She} 
\newcommand{\supplemental}[0]{technical report~\cite{abrtchanTR}}
\newcommand{\twoline}[2]{{\begin{tabular}{c}{{#1}}\\{{#2}}\end{tabular}}}

\begin{document}
\maketitle

\begin{abstract}
Modern applications often operate on data in multiple administrative domains.
In this federated setting, participants may not fully trust each other.
These distributed applications use transactions as a core mechanism for
 ensuring reliability and consistency with persistent data.
However, the coordination mechanisms needed for transactions can both leak
 confidential information and allow unauthorized influence.

%
By implementing a simple attack, we show these side channels can be
 exploited.
However, our focus is on preventing such attacks.
We explore secure scheduling of atomic, serializable transactions in a
 federated setting.
While we prove that no protocol can guarantee security and liveness in all
 settings, we establish conditions for sets of transactions that can safely
 complete under secure scheduling.
Based on these conditions, we introduce \ti{staged commit}, a secure scheduling
 protocol for federated transactions.
This protocol avoids insecure information channels by dividing transactions
 into distinct stages.
We implement a compiler that statically checks code to ensure it meets
 our conditions, and a system that schedules these transactions using the 
 staged commit protocol.
Experiments on this implementation demonstrate that realistic federated 
 transactions can be scheduled securely, atomically, and efficiently.

\end{abstract}

\section{Introduction}
\label{sec:introduction}
Many modern applications are distributed, operating over data from
multiple domains.
Distributed protocols are used by applications to coordinate across
physically separate locations, especially
to maintain data consistency.
However, distributed protocols
can leak confidential information unless carefully designed
 otherwise.

Distributed applications are often structured in terms of \ti{transactions},
 which are atomic groups of operations.
For example, when ordering a book online,
 one or more transactions occur to ensure that the same book is not sold
 twice, and to ensure that the sale of a book and payment transfer happen
 atomically.
Transactions are ubiquitous in modern distributed systems.
Implementations include Google's Spanner~\cite{corbett2013},
 Postgres~\cite{ports2012}, and Microsoft's Azure Storage~\cite{calder2011}.
Common middleware such as Enterprise Java Beans~\cite{java-beans}
 and Microsoft .NET~\cite{NET-transactions} also support
 transactions.

Many such transactions are distributed, involving multiple autonomous
 participants (vendors, banks, etc.). Crucially, these participants may not be equally trusted with all
 data.
Standards such as X/Open XA~\cite{xastandard} aim 
 specifically to facilitate transactions that span multiple systems,
 but none address information leaks inherent to transaction scheduling.

Distributed transaction implementations are often based on the two-phase
 commit protocol (2PC)~\cite{2phase}.
We show that 2PC can create unintentional channels through which private
 information may be leaked, and trusted information may be manipulated.
We expect our results apply to other protocols as well.

There is a fundamental tension between providing strong consistency guarantees 
 in an application and respecting the security requirements of the application's
 trust domains.
This work deepens the understanding of this trade-off and demonstrates 
 that providing both strong consistency and security guarantees, while not 
 always possible, is not a lost cause.

Concretely, we make the following contributions in this paper:
\begin{itemize}
  \item We describe \ti{abort channels}, a new kind of side channel
         through which confidential information can be leaked in transactional
         systems
        (\cref{sec:abortchannels}).
  \item We demonstrate exploitation of abort channels on a distributed system
         (\cref{sec:attackdemo}).
  \item We define an abstract model of distributed systems, transactions, 
         and information flow security (\cref{sec:system}), and
         introduce \ti{\relaxedsecurity}, a noninterference-based
         security model for distributed systems
         (\cref{sec:limited}).
  \item We establish that within this model, it is not possible for any 
         protocol to securely serialize all sets of transactions, even if the 
         transactions are individually secure
         (\cref{sec:impossibility}).
  \item We introduce and prove a sufficient condition for ensuring
         serializable transactions can be securely scheduled
         (\cref{sec:analysis}).
  \item We define the \emph{staged commit} protocol, a novel secure
         scheduling protocol for transactions meeting this condition 
         (\cref{sec:protocols}).
  \item We implement our novel protocol in the Fabric system~\cite{fabric09},
         and extend the Fabric language and compiler to statically ensure
         transactions will be securely scheduled (\cref{sec:implementation}).
  \item We evaluate the expressiveness of the new static checking
         discipline and the runtime overhead of the staged commit
         protocol (\cref{sec:evaluation}).
\end{itemize}
We discuss related work further in \cref{sec:related}, and
conclude in \cref{sec:conclusion}.
\ifreport\else
For brevity, we present proof sketches of the results in this paper;
full proofs can be found in the technical report~\cite{abrtchanTR}.
\fi

\section{Abort Channels}
\label{sec:abortchannels}
Two transactions working with the same data can \ti{conflict} if at least one 
 of them is writing to the data. Typically, this means that one (or
 both) of the transactions has failed and must be \ti{aborted}.
In many transaction protocols, including 2PC, a
participant\footnote{Transaction participants are often processes or network nodes.} involved
in both transactions can abort a failed transaction
by sending an \ti{abort message}
to all other participants in the
failed transaction~\cite{2phase}.
These abort messages can create unintended \ti{abort
 channels}, through which private information can be leaked, and trusted
 information can be manipulated.


An abort message can convey secret information if a participant aborts a 
 transaction otherwise likely to be scheduled, because another participant in 
 the same transaction might deduce something about the aborting participant.
For example, that other participant might guess that the abort is likely caused by the presence 
 of another---possibly secret---conflicting transaction.

Conspirators might deliberately use abort channels to covertly transfer 
 information within a system otherwise believed to be secure.
Although abort channels communicate at most one bit per (attempted) 
 transaction, they could be used as a high-bandwidth covert channel for
 exfiltration of sensitive information.
Current transactional systems can schedule over 100 million
 transactions per second, even at modest system sizes~\cite{dragojevic2015}.
It is difficult to know if abort channels are already being exploited in real 
 systems, but large-scale, multi-user transactional systems such as 
 Spanner~\cite{corbett2013} or Azure Storage~\cite{calder2011} are in
 principle vulnerable.

Abort messages also affect the integrity of transaction scheduling.
An abort typically causes a transaction not to be scheduled.
Even if the system simply retries the transaction until it is scheduled, this
 still permits a participant to control the ordering of transactions, even if
 it has no authority to affect them.
For example, a participant might gain some advantage by ensuring that its own
 transactions always happen after a competitor's.


Transactions can also create channels that leak information based on 
 timing or termination~\cite{atluri2000,Bertino:2001}.
We treat timing and termination channels as outside the scope of this work, to 
 be handled by mechanisms such as timing channel mitigation~\cite{Kopf:Durmuth:CSF2009,azm10,barthe2006}.
Abort channels differ from these previously identified channels in 
 that information leaks via the existence of explicit messages, with
 no reliance on timing other than their ordering.
Timing mitigation does not control abort channels.

\subsection{Rainforest Example}\ifreport \vspace{-3mm}\fi
\label{sec:rainforest}


\begin{figure}
\centering
 \begin{tikzpicture} [scale=1,auto=center]
   \node[inner sep=0pt] (gloriaprime) at (0,2) {};
   \node[inner sep=0pt] (fredprime) at (5,2) {};
   \node[inner sep=0pt] (rainforestprime) at (2.2,2) {};
   \node[inner sep=0pt] (rainforestprime2) at (2.8,2) {};
   \node[inner sep=0pt] (outelprime) at (1.8,0.5) {};
   \node[inner sep=0pt] (bankprime) at (4,0.5) {};
   \node[inner sep=0pt] (outel) at (1,0) 
       {\includegraphics[width=15mm]{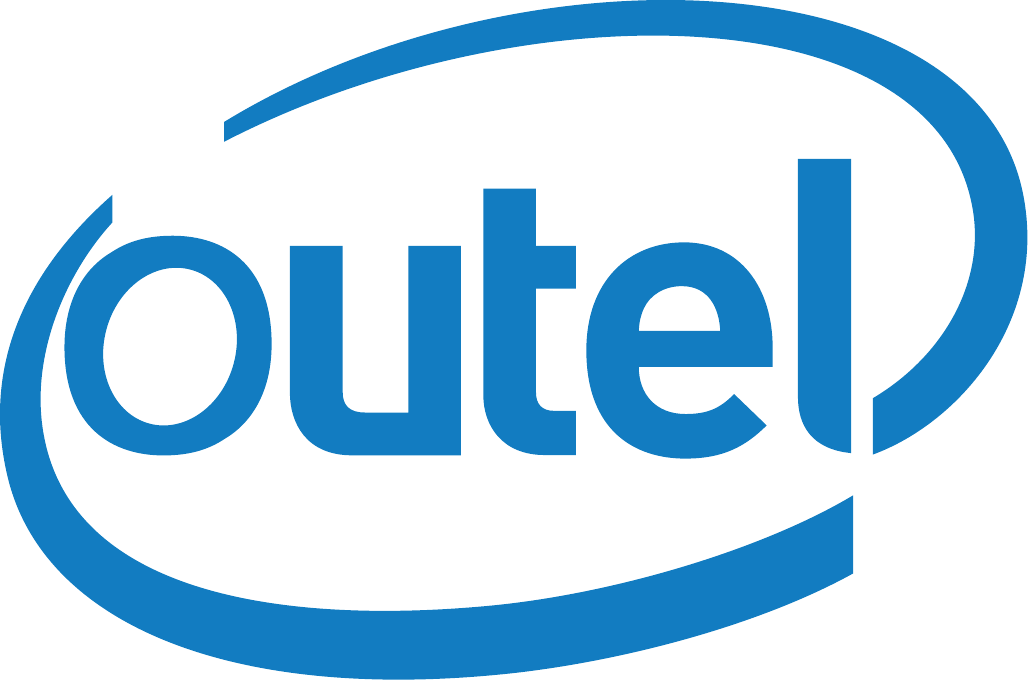}};
   \node[inner sep=0pt] (bank) at (4,0) 
       {\includegraphics[width=15mm]{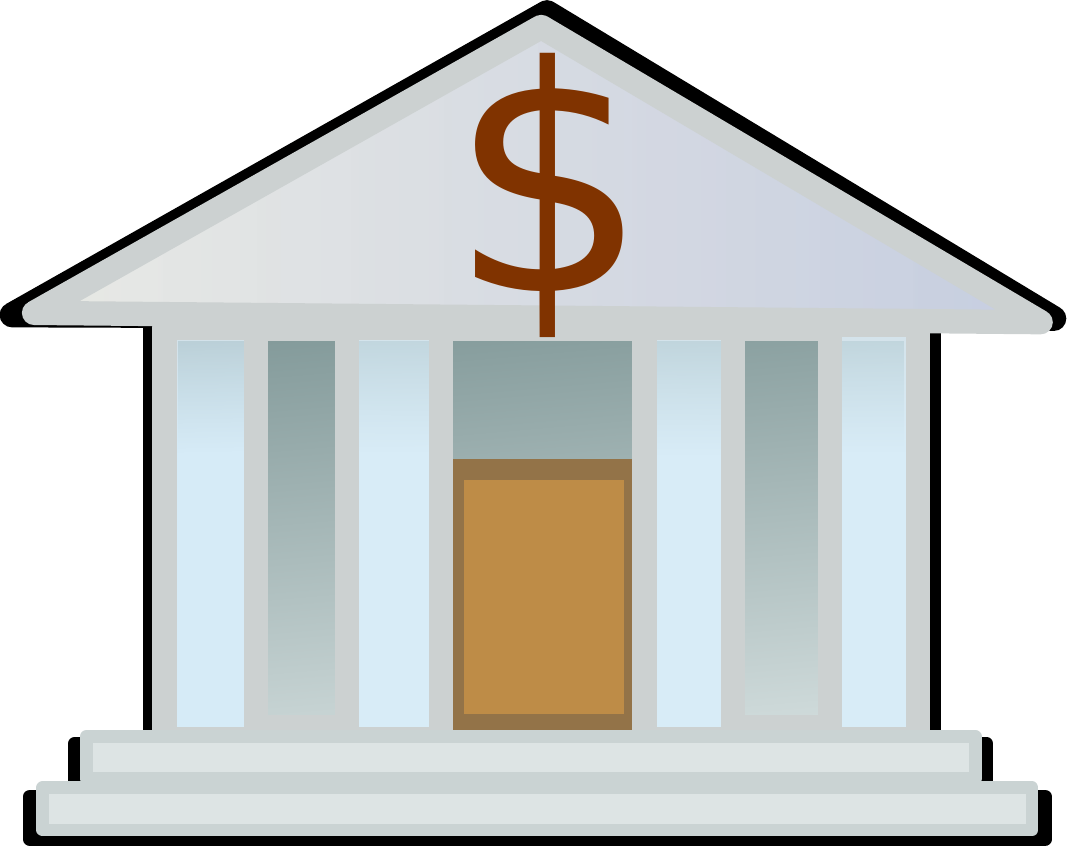}};
   \draw [draw=red,fill=none,dashed,line width=1pt,->,>=open triangle 60]
         (rainforestprime) -- (outel);
   \draw [decoration={text along path,raise=3pt,text={|\small|purchase},
          text align={center}, text color=red},decorate] 
         (outel) -- (rainforestprime);
   \draw [draw=red,fill=none,dashed,line width=1pt,->,>=open triangle 60,label={left:debit account}] 
         (rainforestprime) -- (bank);
   \draw [decoration={text along path,raise=-6pt,text={|\small|\ \ \ \ \ \ \ \ \ debit},
          text align={center}, text color=red},decorate] 
          (rainforestprime) -- (bank);
   \draw [draw=blue,fill=none,line width=1pt,->,>=open triangle 60] 
         (rainforestprime2) -- (outelprime);
   \draw [decoration={text along path,raise=-8pt,text={|\small|purchase},
          text align={left}, text color=blue},decorate] 
         (outelprime) -- (rainforestprime2);
   \draw [draw=blue,fill=none,line width=1pt,->,>=open triangle 60] 
         (rainforestprime2) -- (bankprime);
   \draw [decoration={text along path,raise=3pt,text={|\small|debit},
          text align={center}, text color=blue},decorate] 
         (rainforestprime2) -- (bankprime);
   \node[inner sep=0pt] (rainforest) at (2.5,2) 
   {\begin{tabular}{c}
      Rainforest\\
      \includegraphics[width=15mm]{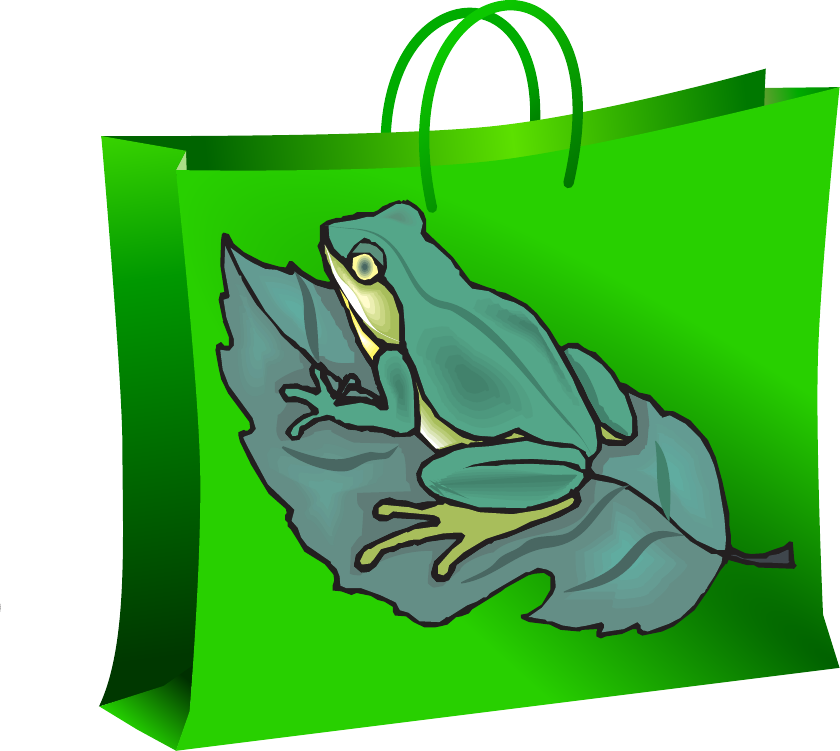}
    \end{tabular}};
   \foreach \from/\to in {gloriaprime/rainforest}  
      \draw [draw=red,fill=none,dashed,line width=1pt,->,>=open triangle 60] 
         (\from) -- (\to);
   \draw [decoration={text along path,raise=3pt,text={|\small|\ \ \ \ order},
          text align={center}, text color=red},decorate] 
         (gloriaprime) -- (rainforest);
   \foreach \from/\to in {fredprime/rainforest}  
      \draw [draw=blue,fill=none,line width=1pt,->,>=open triangle 60] 
         (\from) -- (\to);
   \draw [decoration={text along path,raise=3pt,text={|\small|\ \ \ \ order},
          text align={center}, text color=blue},decorate] 
         (rainforest) -- (fredprime);
   \node[inner sep=0pt] (gloria) at (0,2) 
   {\begin{tabular}{c}
      \includegraphics[width=15mm]{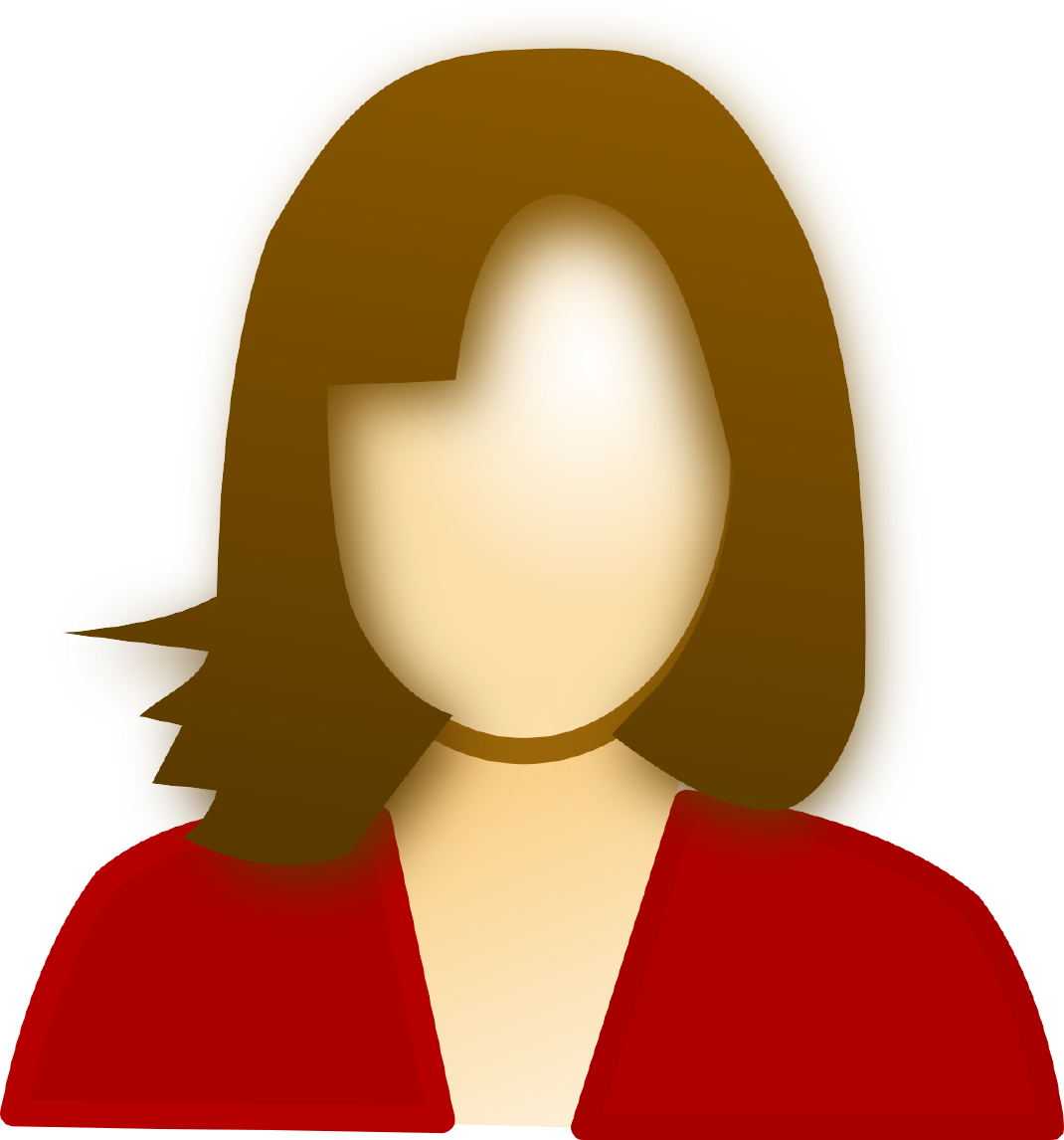}\\
      Gloria
    \end{tabular}};
   \node[inner sep=0pt] (fred) at (5,2) 
   {\begin{tabular}{c}
      \includegraphics[width=15mm]{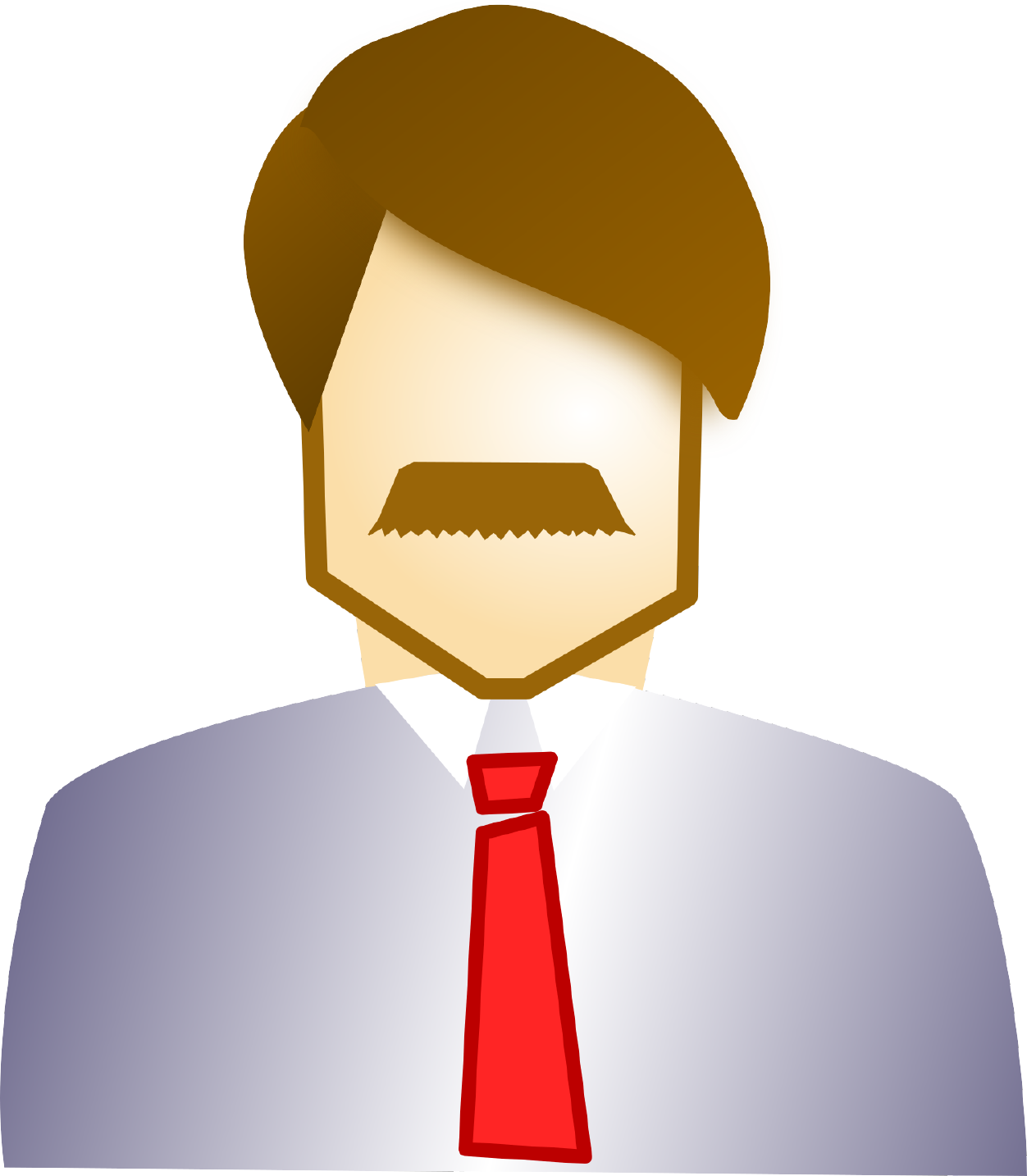}\\
      Fred
    \end{tabular}};
\end{tikzpicture}
\caption{Rainforest example.
  Gloria and Fred each buy an Outel chip via Rainforest's store.
  Gloria's transaction is in {\color{red}red}, dashed arrows; Fred's is 
  in {\color{blue}blue}, solid arrows.
}
\label{fig:rainforest}
\end{figure}

A simple example illustrates how transaction aborts create a
 channel that can leak information.
Consider a web-store application for the fictional on-line retailer Rainforest,   
 illustrated in \cref{fig:rainforest}.
Rainforest's business operates on data from suppliers, customers,  and banks.
Rainforest wants to ensure that it takes money from customers only if the 
 items ordered have been shipped from the suppliers.
As a result, Rainforest implements purchasing using serializable transactions.
Customers expect that their activities do not influence each other, and that 
 their financial information is not leaked to suppliers.
These expectations might be backed by law.

In \cref{fig:rainforest}, Gloria and Fred are both making purchases on 
 Rainforest at roughly the same time.
They each purchase an Outel chip, and pay using their accounts at CountriBank.
\ifreport
\else

\fi
If Rainforest uses 2PC to perform both of these transactions, it is possible for
 Gloria to see an abort when Outel tries to schedule her transaction and Fred's.
The abort leaks information about Fred's purchase at Outel to Gloria.
Alternatively, if Gloria is simultaneously using her bank account in an 
 unrelated purchase, scheduling conflicts at the bank might leak to Outel,   
 which could thereby learn of Gloria's unrelated purchase.

These concerns are about confidentiality, but transactions may also create 
 integrity concerns.
The bank might choose to abort transactions to affect the order in
which Outel sells chips.
Rainforest and Outel may not want the bank to have this power.


\begin{figure}
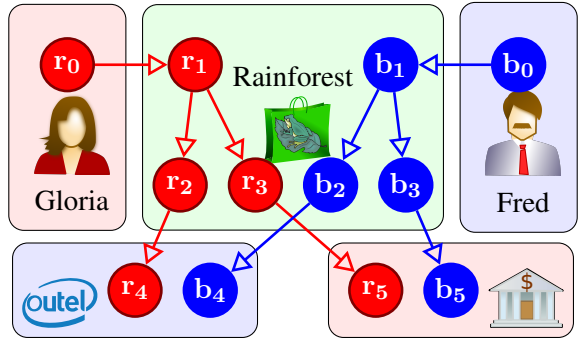

\centering
 \begin{tikzpicture} [scale=1,auto=center]
   \def\labelsize{\normalsize}
   \def\nodesize{7mm}
   \node[rounded corners=2mm, draw, fill=blue!10, minimum height=10mm,minimum width=30mm,text width=30mm,text height=10mm] (outel) at (.9,0) 
     {\includegraphics[width=10mm]{icons/outel.pdf}};
   \node[rounded corners=2mm, draw, fill=red!10, minimum height=10mm,minimum width=30mm,text width=30mm,text height=10mm] (bank) at (5.1,0) 
       {\ \hfill \includegraphics[width=10mm]{icons/Bank.pdf}};
   \node[rounded corners=2mm, draw, fill=green!10, minimum width=40mm,minimum height=30mm] (rainforest) at (3,2.3) 
   {\begin{tabular}{c}
      Rainforest\\
      \includegraphics[width=10mm]{icons/Rainforest.pdf}
    \end{tabular}};
   \node[rounded corners=2mm, draw, fill=red!10, minimum width=10mm,text width=13mm,minimum height=20mm,text height=20mm] (gloria) at (0,2.3) 
   {\begin{tabular}{c}
      \includegraphics[width=10mm]{icons/Gloria.pdf}\\
      Gloria
    \end{tabular}};
   \node[rounded corners=2mm, draw, fill=blue!10, minimum width=10mm,text width=13mm,minimum height=10mm,text height=10mm] (fred) at (6,2.3) 
   {\ \hfill \begin{tabular}{c}
      \includegraphics[width=10mm]{icons/Fred.pdf}\\
      Fred
    \end{tabular}};

  \tikzstyle{every node} = [circle,
      draw={rgb:red,1;black,1},
      inner sep=0pt,
      line width=1pt,
      fill=red,
      minimum size=\nodesize]
  \node (r0) at (0,3) {\color{white}\labelsize$\mathbf{r_0}$};
  \node (r1) at (1.7,3) {\color{white}\labelsize$\mathbf{r_1}$};
  \node (r2) at (1.5,1.4) {\color{white}\labelsize$\mathbf{r_2}$};
  \node (r3) at (2.5,1.4) {\color{white}\labelsize$\mathbf{r_3}$};
  \node (r4) at (0.9,0) {\color{white}\labelsize$\mathbf{r_4}$};
  \node (r5) at (4.1,0) {\color{white}\labelsize$\mathbf{r_5}$};
  \foreach \from/\to in {r0/r1,r1/r2,r1/r3,r2/r4,r3/r5}
      \draw [red,fill=white,->,line width=1pt,>=open triangle 60] (\from) -- (\to);

  \tikzstyle{every node} = [circle,
      inner sep=0pt,
      draw=blue,
      line width=1pt,
      fill=blue,
      minimum size=\nodesize]
  \node (b0) at (6,3) {\color{white}\labelsize$\mathbf{b_0}$};
  \node (b1) at (4.3,3) {\color{white}\labelsize$\mathbf{b_1}$};
  \node (b3) at (4.5,1.4) {\color{white}\labelsize$\mathbf{b_3}$};
  \node (b2) at (3.5,1.4) {\color{white}\labelsize$\mathbf{b_2}$};
  \node (b5) at (5.1,0) {\color{white}\labelsize$\mathbf{b_5}$};
  \node (b4) at (1.9,0) {\color{white}\labelsize$\mathbf{b_4}$};
  \foreach \from/\to in {b0/b1,b1/b2,b1/b3,b2/b4,b3/b5}
      \draw [blue,fill=white,->,line width=1pt,>=open triangle 60] (\from) -- (\to);

\end{tikzpicture}
\caption{
  The events of the transactions in \cref{fig:rainforest}. 
  Gloria's transaction consists of {\color{red}$\mathbf{r_0}$}, 
   {\color{red}$\mathbf{r_1}$}, {\color{red}$\mathbf{r_2}$}, 
   {\color{red}$\mathbf{r_3}$}, {\color{red}$\mathbf{r_4}$}, and 
   {\color{red}$\mathbf{r_5}$}. 
  Bob's consists of 
   {\color{blue}$\mathbf{b_0}$}, {\color{blue}$\mathbf{b_1}$}, 
   {\color{blue}$\mathbf{b_2}$}, {\color{blue}$\mathbf{b_3}$}, 
   {\color{blue}$\mathbf{b_4}$}, and {\color{blue}$\mathbf{b_5}$}.
  Happens-before $\p\before$ relationships are arrows.
  The shaded blocks around events indicate locations, and are labeled with 
   participants from \cref{fig:rainforest}.
}
\label{fig:rainforest-events}
\end{figure}

\subsection{Hospital Example}
\label{sec:hospital-example}
As a second, running example, we use two small programs with an abort channel.
Suppose \Patsy is a trusted hospital employee, running the code in 
 Fig.~\ref{fig:insecure-hospital-patsy} to collect the addresses of
 HIV-positive patients in order to send treatment reminders.
\Patsy runs her transaction on her own computer, which she fully controls, but
 it interacts with a trusted hospital database on another machine.
\Patsy starts a transaction for each patient \texttt p, where
 transaction blocks are indicated by the keyword \texttt{atomic}.
If \texttt p does not have HIV, the transaction finishes immediately.
Fig.~\ref{fig:insecure-hospital-positive} shows the resulting transaction in 
 solid {\color{blue} blue}.
(Events in the transaction are represented as ovals; arrows represent 
 dependencies between transaction events.)
Otherwise, if the patient has HIV, \Patsy's transaction reads the patient's 
 address and prints it (the {\color{blue} blue} transaction 
 in Fig.~\ref{fig:insecure-hospital-positive}, including dashed events).

Suppose \Attacker is another employee at the same hospital, but is not 
 trusted to know each patient's HIV status.
\Attacker is, however, trusted with patient addresses.
Like \Patsy, \Attacker's code runs on \AttackerPossessivePronoun{} own computer,
 which she fully controls, but interacts with the trusted hospital
 database on another machine.
\AttackerPronounCapitalized{} runs the code in 
 Fig.~\ref{fig:insecure-hospital-attacker} to update each patient's address 
 in a separate transaction, resulting in the red transaction in 
 \cref{fig:insecure-hospital-positive}.
When \Attacker updates the address of an HIV-positive patient,
\AttackerPossessivePronoun{} 
 transaction might conflict with one of \Patsy's, and \Attacker would
 observe an abort.
Thus \Attacker can learn which patients are HIV-positive by updating
 each patient's address while \Patsy is checking the patients' HIV
 statuses.
Each time one of \Attacker's transactions aborts, private information leaks: 
 that patient has HIV.

One solution to this problem
is to change \Patsy's transaction:
instead of reading the address only if the patient is HIV positive, \Patsy
 reads every patient's address.
This illustrates a core goal of our work: identifying which programs can be
 scheduled securely.
In Fig.~\ref{fig:secure-hospital-patsy}, lines~3 and~4 of \Patsy's code 
 have been switched.
As Fig.~\ref{fig:secure-hospital-positive} shows, both possible 
 transactions read the patient's address.
Since \Attacker cannot distinguish which of \Patsy's transactions has 
 run, she cannot learn which patients have HIV.

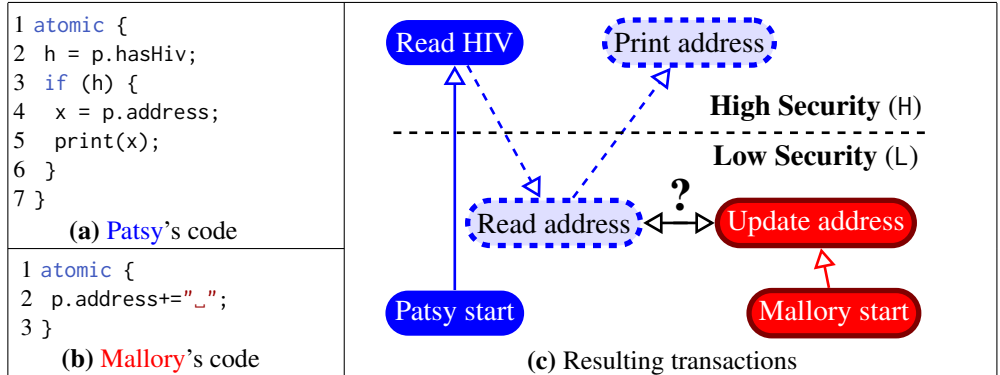
\begin{figure}[t]
  \setlength{\belowcaptionskip}{0pt}
  \setlength{\abovecaptionskip}{0pt}
  \centering
  \begin{tabular}{|l | r|}
    \hline
      \hspace{-6mm}
      \ifreport
	\begin{subfigure}{0.27\textwidth}
      \else
	\begin{subfigure}{0.17\textwidth}
      \fi
        \vspace{-1mm} 
        \begin{lstlisting}[gobble=10]
          atomic {
           h = p.hasHiv;
           if (h) {
            x = p.address;
            print(x);
           }    
          }
        \end{lstlisting}
        \vspace{-4mm}
        \caption{{\color{blue}\Patsy}'s code}
        \label{fig:insecure-hospital-patsy}
        \vspace{0mm}
      \end{subfigure}
  &
   \multirow{2}{\ifreport 0.5 \else 0.32 \fi\textwidth}{
     \hspace{-4mm}
    \setcounter{subfigure}{2}
      \ifreport
	\begin{subfigure}{0.5\textwidth}
      \else
	\begin{subfigure}{0.32\textwidth}
      \fi
      \ifreport
        \vspace{-12mm}
      \else
        \vspace{-11mm}
      \fi
      \centering
      \ifreport
        \begin{tikzpicture} [scale=1.2,auto=left]
      \else
        \begin{tikzpicture} [scale=1,auto=center]
      \fi
        \def\labelsize{\normalsize}
        \def\nodesize{7mm}
        \node[rounded corners=3mm, fill=blue, minimum height=6mm] (b0) at (0,0) 
          {\color{white} \Patsy start};
        \node[rounded corners=3mm, fill=blue, minimum height=6mm] (b1) at (0,3) 
          {\color{white} Read HIV};
        \node[rounded corners=3mm, fill={rgb:blue,1;white,7}, minimum height=6mm,line width=2pt,draw=blue,dashed] (b2) at (1.1,1) 
          {\color{black} Read address};
        \node[rounded corners=3mm, fill={rgb:blue,1;white,7}, minimum height=6mm,line width=2pt,draw=blue,dashed] (b3) at (2.6,3) 
          {\color{black} Print address};
        \foreach \from/\to in {b0/b1}
           \draw [blue,->,fill=white,line width=1pt,>=open triangle 60] (\from) -- (\to);
        \foreach \from/\to in {b1/b2,b2/b3}
           \draw [blue,dashed,->,fill=white,line width=1pt,>=open triangle 60] (\from) -- (\to);
        \node[rounded corners=3mm, draw={rgb:red,1;black,1}, line width=2pt, fill=red,
              minimum height=6mm] (r0) at (4.2,0) 
          {\color{white} \Attacker start};
        \node[rounded corners=3mm, draw={rgb:red,1;black,1}, line width=2pt, fill=red,
              minimum height=6mm] (r1) at (4,1) 
          {\color{white} Update address};
        \draw [red,->,fill=white,line width=1pt,>=open triangle 60] (r0) -- (r1);
        \node[inner sep=0pt,label={\LARGE\tb?}] (order) at (2.5,1) {};
        \node[inner sep=0pt] (order0) at (2.6,1) {};
        \node[inner sep=0pt] (order2) at (2.4,1) {};
        \foreach \from/\to in {order0/b2,order2/r1}
           \draw [black,->,fill=white,line width=1pt,>=open triangle 60] (\from) -- (\to);
        \node[inner sep=0pt] (left) at (-0.7,2) {};
        \node[inner sep=0pt] (right) at (5.25,2) {};
        \node[inner sep=0pt,label={\tb{High Security} (\texttt H)}] (high) at (4,2) {};
        \node[inner sep=0pt,label={below:\tb{Low Security} (\texttt L)}] (high) at (4,2) {};
        \draw [black,-,dashed,line width=1pt] (left) -- (right);
      \end{tikzpicture}
     \ifreport
       \vspace{-1mm}
     \else
       \vspace{-5mm}
     \fi
    \caption{Resulting transactions}
      \label{fig:insecure-hospital-positive}
    \end{subfigure}
   }

      \\\cline{1-1}
 
      \setcounter{subfigure}{1}
      \hspace{-6mm}
      \ifreport
	\begin{subfigure}{0.27\textwidth}
      \else
	\begin{subfigure}{0.17\textwidth}
      \fi
        \vspace{-1mm}
        \centering
        \begin{lstlisting}[gobble=10]
          atomic {
           p.address+=" ";
          }
        \end{lstlisting}
        \vspace{-4mm}
        \caption{{\color{red}\Attacker}'s code}
        \label{fig:insecure-hospital-attacker}
        \vspace{0mm}
      \end{subfigure}
      &
  \\\hline
  \end{tabular}
  \setlength{\belowcaptionskip}{\correctbelowcaptionskip}
  \caption{
    Insecure hospital scenario.
    {\color{blue}\Patsy} runs a program (\ref{fig:insecure-hospital-patsy}) 
     for each patient \texttt p.
    If \texttt p has HIV (which is private information), she 
     prints out \texttt p's address for her records.
    The resulting transaction takes one of two forms. 
    Both begin with the event {\color{blue}\Patsy start}.
    If \texttt p is HIV negative, the transaction ends with {\color{blue}Read 
     HIV}.
    Otherwise, it includes the {\color{blue}blue} events with dashed outlines.
    Meanwhile, {\color{red}\Attacker} updates the \texttt p's (less secret) 
     address (\ref{fig:insecure-hospital-attacker}), resulting in the 
     transaction with {\color{red}red}, solid-bordered events.
     This conflicts with {\color{blue}\Patsy's transaction}, 
     requiring the system to order the {\color{red}update} and the 
     {\color{blue}read}, exactly when \texttt p has HIV
     (``\tb?'' in~\ref{fig:insecure-hospital-positive}).
  }
  \ifreport
  \else
    \vspace{-6mm}
  \fi
\label{fig:insecure-hospital}
\end{figure}

\subsection{Attack Demonstration}\ifreport\vspace{-3mm}\fi
\label{sec:attackdemo}
Using code resembling~\cref{fig:insecure-hospital},
we implemented the attack described in our hospital example 
 (\cref{sec:hospital-example}) using the Fabric distributed 
 system~\cite{fabric-release03,fabric09}.
We ran nodes representing {\Patsy} and {\Attacker}, and a storage node
for the patient records.

To improve the likelihood of {\Attacker} conflicting with {\Patsy}
(and thereby receiving an abort), we had 
 {\Patsy} loop roughly once a second, continually reading the address of
 a single patient after verifying their HIV-positive status.
Meanwhile, {\Attacker} attempted to update the patient's address
 with approximately the same frequency as {\Patsy}'s transaction.

Like many other distributed transaction systems,
Fabric uses two-phase commit. {\Attacker}'s window of
opportunity for receiving an abort exists between the two phases of
{\Patsy}'s commit, which ordinarily involves a network round trip.
However, both nodes were run on a single computer.
To model a cloud-based server, we
 simulated a 100~ms network delay between {\Patsy} and the storage
 node.

\ifreport
Getting this to work was challenging, because Fabric caches its
objects optimistically. When {\Attacker} updates the patient's
address, it would invalidate {\Patsy}'s cached copy, causing
{\Patsy}'s next transaction to abort and retry.
Furthermore, Fabric implements an exponential back-off algorithm for
retrying aborted transactions.
As a result, we had to carefully tune the transaction frequencies to
prevent {\Attacker} from starving out {\Patsy}.
\fi

We ran this experiment for 90 minutes.
During this time, {\Attacker} received an abort
 roughly once for every 20 transactions
 {\Patsy} attempted.
As a result, approximately every 20 seconds, {\Attacker} learned 
 that a patient had HIV.
In principle, many such attacks could be run in parallel, so
this should be seen as a minimal, rather than a maximal, rate of 
 information leakage for this setup.

As described later,
our modified Fabric compiler~(\cref{sec:implementation}) correctly
rejects \Patsy's code.
We amended \Patsy's code to reflect \cref{fig:secure-hospital}, and our
implementation of 
 the staged commit protocol~(\cref{sec:protocols}) was able to 
 schedule the transactions without leaking information.
{\Attacker} was no more or less likely to receive aborts regardless 
 of whether the patient had HIV.

\section{System Model}
\label{sec:system}
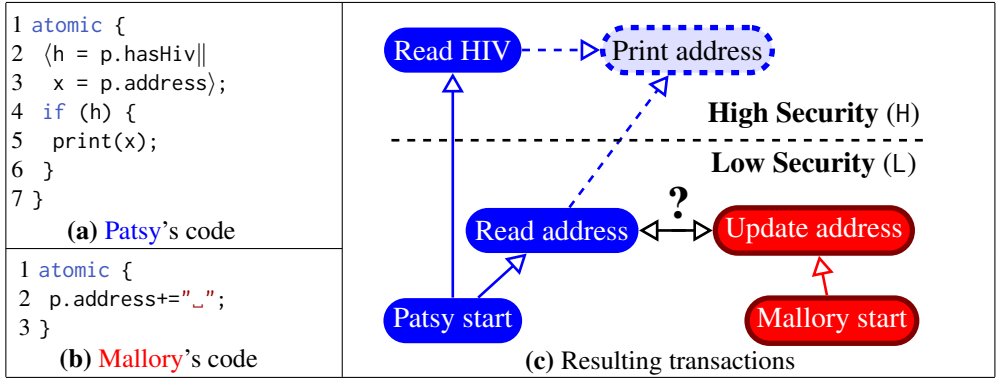
\begin{figure}[t]
  \centering
  \setlength{\belowcaptionskip}{0pt}
  \setlength{\abovecaptionskip}{0pt}
  \begin{tabular}{|l | r|}
    \hline
      \hspace{-6mm}
      \begin{subfigure}{\ifreport 0.27 \else 0.17 \fi\textwidth}
        \vspace{-1mm} 
        \centering
        \lstset{
            literate={〈}{{\ensuremath\langle}}1
                     {〉}{{\ensuremath\rangle}}1
                     {‖}{{\ensuremath\parallel}}1
            }
        \begin{lstlisting}[gobble=10]
          atomic {
           〈h = p.hasHiv‖
            x = p.address〉;
           if (h) {
            print(x);
           }    
          }
        \end{lstlisting}
        \vspace{-4mm}
      \caption{{\color{blue}\Patsy}'s code}
        \label{fig:secure-hospital-patsy}
        \vspace{0mm}
      \end{subfigure}

  &
   \multirow{2}{\ifreport 0.5 \else 0.32 \fi\textwidth}{
    \hspace{-4mm}
    \setcounter{subfigure}{2}
    \begin{subfigure}[b]{\ifreport 0.5 \else 0.32 \fi\textwidth}
    \vspace{-11mm}
      \centering
      \ifreport
        \begin{tikzpicture} [scale=1.2,auto=left]
      \else
        \begin{tikzpicture} [scale=1,auto=center]
      \fi
        \def\labelsize{\large}
        \def\nodesize{7mm}
        \node[rounded corners=3mm, fill=blue, minimum height=6mm] (b0) at (0,0) 
          {\color{white} \Patsy{} start};
        \node[rounded corners=3mm, fill=blue, minimum height=6mm] (b1) at (0,3) 
          {\color{white} Read HIV};
        \node[rounded corners=3mm, fill=blue, minimum height=6mm] (b2) at (1.1,1) 
          {\color{white} Read address};
        \node[rounded corners=3mm, fill={rgb:blue,1;white,7}, minimum height=6mm,line width=2pt,draw=blue,dashed] (b3) at (2.6,3) 
          {\color{black} Print address};
        \foreach \from/\to in {b0/b1,b0/b2}
           \draw [blue,->,fill=white,line width=1pt,>=open triangle 60] (\from) -- (\to);
        \foreach \from/\to in {b1/b3,b2/b3}
           \draw [blue,dashed,->,fill=white,line width=1pt,>=open triangle 60] (\from) -- (\to);
        \node[rounded corners=3mm, draw={rgb:red,1;black,1}, line width=2pt, fill=red,
              minimum height=6mm] (r0) at (4.2,0) 
          {\color{white} \Attacker{} start};
        \node[rounded corners=3mm, draw={rgb:red,1;black,1}, line width=2pt, fill=red,
              minimum height=6mm] (r1) at (4,1) 
          {\color{white} Update address};
        \draw [red,->,fill=white,line width=1pt,>=open triangle 60] (r0) -- (r1);
        \node[inner sep=0pt,label={\LARGE\tb?}] (order) at (2.5,1) {};
        \node[inner sep=0pt] (order0) at (2.6,1) {};
        \node[inner sep=0pt] (order2) at (2.4,1) {};
        \foreach \from/\to in {order0/b2,order2/r1}
           \draw [black,->,fill=white,line width=1pt,>=open triangle 60] (\from) -- (\to);
        \node[inner sep=0pt] (left) at (-0.7,2) {};
        \node[inner sep=0pt] (right) at (5.25,2) {};
        \node[inner sep=0pt,label={\tb{High Security} (\texttt H)}] (high) at (4,2) {};
        \node[inner sep=0pt,label={below:\tb{Low Security} (\texttt L)}] (high) at (4,2) {};
        \draw [black,-,dashed,line width=1pt] (left) -- (right);
      \end{tikzpicture}
    \vspace{\ifreport -2 \else -5 \fi mm}
    \caption{Resulting transactions}
      \label{fig:secure-hospital-positive}
    \end{subfigure}
   }

      \\\cline{1-1}

      \hspace{-6mm}
      \setcounter{subfigure}{1}
      \begin{subfigure}{\ifreport 0.27 \else 0.17 \fi\textwidth}
        \vspace{-1mm}
        \centering
        \begin{lstlisting}[gobble=10]
          atomic {
           p.address+=" ";
          }
        \end{lstlisting}
        \vspace{-4mm}
        \caption{{\color{red}\Attacker{}}'s code}
        \label{fig:secure-hospital-attacker}
        \vspace{0mm}
      \end{subfigure}
      &
  \\\hline
  \end{tabular}
  \setlength{\belowcaptionskip}{\correctbelowcaptionskip}
  \caption{
    Secure hospital scenario.
    A secure version of \cref{fig:insecure-hospital}, in which lines 3 
    and 4 of {\color{blue}\Patsy{}}'s code (\ref{fig:insecure-hospital-patsy}) 
    are switched, and the resulting lines 2 and 3 can be run in parallel ($\an{\ \parallel\ }$).
    Thus the transaction reads \texttt p's address regardless of whether 
     \texttt p has HIV, and so {\color{red}\Attacker{}} cannot distinguish
     which form {\color{blue}\Patsy}'s transaction takes.
  }
  \vspace{-5mm}
\label{fig:secure-hospital}
\end{figure}

We introduce a formal, abstract system model that serves as our
framework for developing protocols and proving their security
properties.
Despite its simplicity, the model captures the necessary features of 
distributed transaction systems and protocols.
As part of this 
model, we define what it means for transactions to be
serializable and what it means for a protocol to serialize
transactions both correctly and securely.

\subsection{State and Events}
\label{sec:state}

Similarly to Lamport~\cite{Lamport78}, we define a \ti{system state}
to include a finite set of \ti{events}, representing a history of the
system up to a moment in time.
%
%
%
An event (denoted $e$) is an atomic native action that takes place at a
 \ti{location}, which can be thought of as a physical computer on the network.
Some events may represent read operations (``the variable $x$ had the
value 3''), or write operations (``2 was written into the variable 
 $y$''). 
In \cref{fig:insecure-hospital,fig:secure-hospital}, for
 example, events are represented as ovals, and correspond to lines of code.

Also part of the system state is a causal ordering on events.
Like Lamport's causality~\cite{Lamport78}, the ordering
describes when one event $e_1$ causes another event $e_2$. In this
case, we say $e_1$ \ti{happens before} $e_2$, written as $e_1\before
e_2$.  This relationship would hold if, for example, $e_1$ is the
sending of a message, and $e_2$ its receipt.
The ordering $\p\before$ is a strict partial order: irreflexive,
asymmetric, and transitive. Therefore, $e_1\before e_2$ and $e_2\before e_3$
together imply $e_1\before e_3$.

The arrows in 
\cref{fig:insecure-hospital,fig:secure-hospital,fig:rainforest-events}
 show happens-before relationships for the transactions involved.
 
\subsection{Information Flow Lattice}

We extend Lamport's model by assigning to each event $e$ a \ti{security
label}, 
 written $\ell\p e$, which defines the confidentiality and integrity
 requirements of the event.
Events are the most fine-grained unit of information in our model, so there is no
 distinction between the confidentiality of an event's
 \ti{occurrence} and that of its \ti{contents}.
 \TM{It might be worth pointing out to the reader why this is
 reasonable, using the attack as an example of knowing the event
 happened being treated as equivalent to knowing what the event's
 content was?  Simulating a new reader, this sentence reads a bit
 ``surprising'' right now if you're not thinking all that hard about
 it.}
 \TMreply{I think we addressed part of this, I'm leaving the comment in
 case anyone would like to reconsider my suggestion of using the attack
 as an example of why it's a reasonable statement.}
Labels in our model are similar to high and low event sets~\cite{Roscoe95,cs08}.
In \cref{fig:insecure-hospital,fig:secure-hospital},
two security labels, \texttt{High} and \texttt{Low} (\texttt H and 
 \texttt L for short),
are represented by the events' positions relative to the dashed
 line.

For generality, we assume that labels
 are drawn from a lattice~\cite{denning-lattice}, depicted in
 \cref{fig:lattice}.
Information is only permitted to flow upward in the lattice.
We write ``$\ell\p{e_1}$ is below $\ell\p{e_2}$'' as 
 $\ell\p{e_1} \less \ell\p{e_2}$, meaning it is secure
 for the information in $e_1$ to flow to $e_2$.

For instance, in \cref{fig:insecure-hospital}, information should not 
 flow from any events labeled \texttt H to any labeled \texttt L.
Intuitively, we don't want secret information to determine any non-secret 
 events, because unauthorized parties might learn something secret.
However, information can flow in the reverse direction: reading the patient's  
 address (labeled \texttt L) can affect \Patsy's printout
 (labeled \texttt H): \texttt L ⊑ \texttt H.

\ifreport
The join $\p\join$ of two labels represents their least upper bound:
 $\ell_1\less\p{\ell_1\join\ell_2}$ and $\ell_2\less\p{\ell_1\join\ell_2}$.
The meet ($\meet$) of two labels represents their greatest lower bound:
$\p{\ell_1\meet\ell_2}\less\ell_1$ and $\p{\ell_1\meet\ell_2}\less\ell_2$.
\fi

Like events, each location has a label, representing a limit on events with 
 which that location can be trusted.
No event should have more integrity than its location.
Similarly, no event should be too secret for its location to know.
Thus, in~\cref{fig:lattice}, only events $\ti{to the left of}$ a
location's
 label (i.e., region C in the figure) may take place at that location.
\label{sec:lattice}

For example, consider Gloria's payment event at CountriBank in the Rainforest
 example  \cref{fig:rainforest}.
This event ({\color{red}$\mathbf{r_5}$} in \cref{fig:rainforest-events})
 represents money moving from Gloria's account to Outel's. 
The label $\ell$ of {\color{red}$\mathbf{r_5}$} should not have any more integrity
 than CountriBank itself, since the bank controls {\color{red}$\mathbf{r_5}$}.
Likewise, the bank knows about {\color{red}$\mathbf{r_5}$}, so $\ell$ cannot be
 more confidential than the CountriBank's label.
This would put $\ell$ to the \ti{left} of the label representing CountriBank in 
 the lattice of \cref{fig:lattice}.

Our prototype implementation of secure transactions is built using the Fabric
 system~\cite{fabric09}, so the lattice used in the implementation is based on 
 the Decentralized Label Model (DLM)~\cite{ml-tosem}. 
However, the results of this paper are independent of the lattice used.

\begin{figure}[t]
\centering
\begin{tikzpicture}
 [scale=\ifreport .4 \else .37 \fi,auto=left, >=latex, every node/.style=]
   \usepgflibrary{arrows}
 \node[scale=0.015,label={below:\twoline{public}{trusted}}] (d) at (11,5) {};
 \node[scale=0.015,label={above:\twoline{secret}{untrusted}}] (u) at (11,15)  {};
 \node[scale=0.015,label={left:\twoline{public}{untrusted}}] (l) at (6,10)  {};
 \node[scale=0.015,label={right:\twoline{secret}{trusted}}] (r) at (16,10)  {};
   \node (cd) at (11,8.3) {};
   \node (cu) at (11,11.7) {};
   \node (il) at (10,4.5) {};
   \node (ih) at (5.5,9) {};
   \node (al) at (12,4.5) {};
   \node (ah) at (16.5,9) {};
   
   \node[scale=0.015] (hi) at (12.5,6.5) {};
   \node[scale=0.015] (li) at (7.5,11.5) {};
   \node[scale=0.015] (hc) at (12.5,13.5) {};
   \node[scale=0.015] (lc) at (7.5,8.5) {};
   \node[circle,scale=0.5,draw,fill=black] (dot) at (9,10) {};
   \node (ld) at (7.5,10) {C};
   \node (td) at (9,11.5) {A};
   \node (bd) at (9,8.3) {B};
   \tikzstyle{every node} = [scale=1,auto=left,every node/.style=]
   \tikzstyle{every node} = [scale=1,auto=left,every node/.style=]
   
    \path[dashed] (hi) edge (li);
    \path[dashed] (hc) edge (lc);
    \path (l) edge node[above] {} (u);
    \path (l) edge node[above] {} (d);
    \path (r) edge node[above] {} (u);
    \path (r) edge node[above] {} (d);
    \path[->] (cd) edge node[right] {$\less$} (cu) ;
    \path[<-] (il) edge node[left] {} (ih);
    \path[->,decoration={text along path,raise=-9pt,
          text={{integrity}},
          text align={center}, text color=black},decorate] (ih) -- (il) ;
    \path[->] (al) edge node[right] {} (ah) ;
    \path[->,decoration={text along path,raise=-9pt,
          text={{confidentiality}},
          text align={center}, text color=black},decorate] 
          (al) -- (ah) ;
 \end{tikzpicture}
\caption{
Security lattice:
The dot represents a label in the lattice, and the dashed lines divide the 
 lattice into four quadrants relative to this label.
If the label represents an event, then only events with labels in quadrant 
 $B$   may influence this event, and this event may only influence events with 
 labels in quadrant $A$.
If the label represents a location, then only events with labels in quadrant 
 $C$ may occur at that location.
}
\label{fig:lattice}
\end{figure}
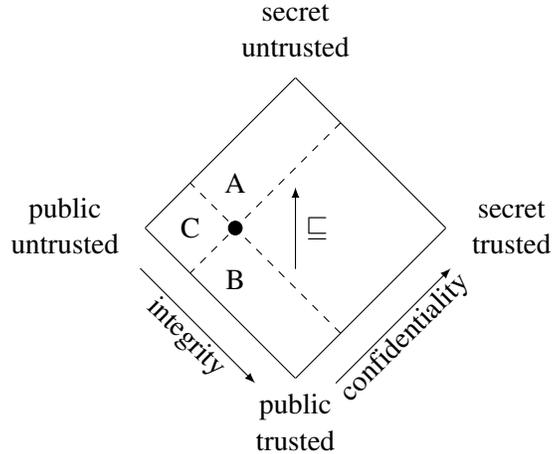

\newcommand{\hospitalfull}[4] {{
      \begin{tikzpicture} [scale=\ifreport 1.1 \else 0.9 \fi,auto=center]
 \ifthenelse{#1>0}{
    \scope
    \def\lw{1pt}
        \node[inner sep=0pt] (left) at (-2.2,7) {};
        \node[inner sep=0pt] (right) at (6.2,7) {};
        \node[inner sep=0pt,label={\large\tb{High Security}}] (high) at (4.7,7) {};
        \node[inner sep=0pt,yshift=2pt,label={below:\large\tb{Low Security}}] (high) at (4.7,7) {};
        \draw [black,-,dashed,line width=1pt] (left) -- (right);
    }{
    \useasboundingbox (0,0) rectangle (1.3cm,1.1cm);
    \scope[transform canvas={scale=0.2}]
    \def\lw{3pt}
  }
      \def\labelsize{\large}
      \def\nodesize{7mm}
      \ifthenelse{#2>0}{
          \node[rounded corners=3mm, fill=blue, minimum height=6mm] (b0) at (-1.5,0) 
            {\color{white} \Patsy\ start};
      }{}
      \ifthenelse{#2>1}{
        \node[rounded corners=3mm, fill=blue, minimum height=6mm] (b2) at (-0.2,2.1) 
          {\color{white} Read address};
        }{}
      \ifthenelse{#2>2}{
        \node[rounded corners=3mm, fill=blue, minimum height=6mm] (b1) at (-1.5,7.7) 
          {\color{white} Read HIV};
        }{}
      \ifthenelse{#2>3}{
        \node[rounded corners=3mm, fill=blue, minimum height=6mm] (b3) at (1,8.7) 
          {\color{white} Print address};
        }{}
      \ifthenelse{#2>3}{\def\r{b0/b2,b0/b1,b1/b3,b2/b3};}{
         \ifthenelse{#2>2}{\def\r{b0/b2,b0/b1};}{
           \ifthenelse{#2>1}{\def\r{b0/b2};}{
             \def\r{};}}}
        \foreach \from/\to in \r
           \draw [blue,->,fill=white,line width=1pt,>=open triangle 60] (\from) -- (\to);

      \ifthenelse{#3>0}{
        \node[rounded corners=3mm, draw={rgb:red,1;black,1}, line width=2pt, fill=red,
              minimum height=6mm] (r0) at (5.2,0) 
          {\color{white} \Attacker\ start};
        }{}
      \ifthenelse{#3>1}{
        \node[rounded corners=3mm, draw={rgb:red,1;black,1}, line width=2pt, fill=red,
              minimum height=6mm] (r1) at (5.2,5) 
          {\color{white} Update address};
        \draw [red,->,fill=white,line width=1pt,>=open triangle 60] (r0) -- (r1);
      }{}

    \tikzstyle{every node} = [rounded corners=2.5mm,
                              draw=black,
                              line width=2pt,
                              fill=none,
                              minimum height=5mm]
    \ifthenelse{#4>0}{\node (p0) at (1.7,1.2) {{\color{blue}\Patsy}
    {\color{DarkOrchid} acquires lock}};}{}
    \ifthenelse{#4>1}{\node (p1) at (1.7,3.0) {{\color{blue}\Patsy}
    {\color{DarkOrchid} releases lock}}; }{}
    \ifthenelse{#4>2}{\node (p2) at (3.4,4) {{\color{red}\Attacker}
    {\color{DarkOrchid} acquires lock}};}{}
    \ifthenelse{#4>3}{\node (p3) at (3.4,6.1) {{\color{red}\Attacker}
    {\color{DarkOrchid} releases lock}};}{}

    \ifthenelse{#4>3}{\def\pro{b0/p0,p0/b2,b3/p1,p0/p1,p1/p2,p2/r1,r1/p3,p2/p3,r0/p2,b2/p1};}{
    \ifthenelse{#3>1}{\def\pro{b0/p0,p0/b2,b3/p1,p0/p1,p1/p2,r0/p2,p2/r1};}{
    \ifthenelse{#4>2}{\def\pro{b0/p0,p0/b2,b3/p1,p0/p1,p1/p2,r0/p2};}{
      \ifthenelse{#4>1}{\def\pro{b0/p0,p0/b2,b3/p1,p0/p1};}{
        \ifthenelse{#2>1}{\def\pro{b0/p0,p0/b2};}{
          \ifthenelse{#4>0}{\def\pro{b0/p0};}{
            \def\pro{};}}}}}}
    \foreach \from/\to in \pro 
      \draw [black,->,line width=\lw,>=open triangle 60] (\from) -- (\to);
  \endscope
\end{tikzpicture}
}}

\newcommand{\transactionsdiagram}[8] {{
 \ifthenelse{#8>0}{
    \newcommand{\nodesize}{1.7em}
    \begin{tikzpicture} [scale=.5,auto=center]
      \def\nodefont{\normalsize}
    \scope
    }{
    \newcommand{\nodesize}{10mm}
    \begin{tikzpicture} [scale=1,auto=center]
      \def\nodefont{\LARGE}
    \useasboundingbox (0,0) rectangle (\ifreport 1 \else 1.3 \fi cm,1.1cm);
    \scope[transform canvas={scale=\ifreport 0.17 \else 0.2 \fi}]
  }
  \tikzstyle{every node} = [circle,
      draw={rgb:red,1;black,1},inner sep=0pt,
      line width=1pt,
      fill=red,
      minimum size=\nodesize]
    \ifthenelse{#1>0}{\node (n1) at (0,0) {\color{white}\nodefont$\mathbf{r_0}$};}{}
    \ifthenelse{#2>0}{\node (n2) at (2,1) {\color{white}\nodefont$\mathbf{r_1}$};}{}
    \ifthenelse{#3>0}{\node (n3) at (0,5) {\color{white}\nodefont$\mathbf{r_2}$};}{}
    
    \ifthenelse{#3>0}{\def\r{n1/n2,n2/n3};}{
       \ifthenelse{#2>0}{\def\r{n1/n2};}{
          \def\r{};}}
  \ifthenelse{#8>0}{\def\lw{1pt}}{\def\lw{3pt}}
  \foreach \from/\to in \r 
    \draw [draw={rgb:red,1;black,1},
             fill=none,
             line width=\lw,
             ->,
             >=open triangle 60] (\from) -- (\to);
    \tikzstyle{every node} = [circle,draw=blue,line width=1pt,fill=blue,inner sep=0pt,minimum size=\nodesize]
    \ifthenelse{#4>0}{\node (n4) at (7,0) {\color{white}\nodefont$\mathbf{b_0}$};}{}
    \ifthenelse{#5>0}{\node (n5) at (5,4) {\color{white}\nodefont$\mathbf{b_1}$};}{}
    \ifthenelse{#6>0}{\node (n6) at (7,5) {\color{white}\nodefont$\mathbf{b_2}$};}{}
    \ifthenelse{#6>0}{\def\b{n4/n5,n5/n6};}{
       \ifthenelse{#5>0}{\def\b{n4/n5};}{
          \def\b{};}}
    \foreach \from/\to in \b 
      \draw [blue,fill=none,line width=\lw,->,>=open triangle 60] (\from) -- (\to);
    \ifthenelse{#8>0}{\def\peventlw{2pt}}{\def\peventlw{6pt}}
    \tikzstyle{every node} = [circle,
                              draw=black,
                              line width=\peventlw,
                              fill=none,inner sep=0pt,
                              minimum size=\nodesize]
    \ifthenelse{#7>0}{\node (n7) at (3.5,2.5) {\color{black}\nodefont$\mathbf{p}$};
    \ifthenelse{#5>0}{\def\pro{n2/n7,n7/n5};}{
       \def\pro{n2/n7};}
    \foreach \from/\to in \pro 
      \draw [black,->,line width=\lw,>=open triangle 60] (\from) -- (\to);}{}
  \endscope
\end{tikzpicture}
}}

\subsection{Conflicts}
\label{sec:conflict}
Two events in different transactions may \ti{conflict}.
This is a property inherent to some pairs of events.
Intuitively, conflicting events are events that must be ordered for
data to be consistent.
For example,
if $e_1$ represents reading variable $x$, and $e_2$ represents 
 writing $x$, then they conflict, and furthermore, the value read and the value
 written establish an ordering between the events.
Likewise, if two events both write variable $x$, they conflict, and the system must
 decide their ordering because it affects future reads of $x$.

In our hospital example 
 (\cref{fig:insecure-hospital,fig:secure-hospital}), the events
 {\color{blue}Read address} and {\color{red}Update address} conflict.
Specifically, the value read will change depending on whether it is read before 
 or after the update.
Thus for any such pair of events, there is a happens-before
$\p\before$ ordering between them, in one direction or the other.

We assume that conflicting events have the same label. This assumption
is intuitive in the case of events that are reads and writes to the
same variable (that is, storage location). Read and write operations
in separate transactions could have occurred in either order, so the
happens-before relationship between the read and write events cannot
be predicted in advance.

Our notion of \ti{conflict} is meant to describe direct interaction between
transactions. Hence, we also assume any conflicting events happen at the same 
location.

\subsection{Serializability and \TransactionSecurity}
\label{sec:serializability}
Traditionally a transaction is modeled as a set of reads and writes to 
 different objects~\cite{Papa79}.
We take a more abstract view, and model a transaction as a set of 
 events that arise from running a piece of code.
Each transaction features a \ti{start event}, representing the decision 
 to execute the transaction's code.
Start events, by definition, happen before all others in the transaction.
Multiple possible transactions can feature the same start event: the complete
 behavior of the transaction's code is not always determined when it
 starts executing, and may
 depend on past system events.

\cref{fig:secure-hospital-positive}
 shows two possible transactions, in {\color{blue}blue}, that can
 result from running the secure version of {\color{blue}\Patsy}'s code.
They share the three events in solid {\color{blue}blue}, including the
start event ({\color{blue}\Patsy\ start}); one transaction contains a
fourth event, {\color{blue}Print address}.
The figure also shows in {\color{red}red} the transaction resulting from {\color{red}\Attacker}'s
 code.
\cref{fig:red-blue-transactions} is a more abstract example, in which
 {\color{red}$\mathbf{r_0}$} is the start event of transaction
 {\color{red}$\mathbf{R}$}, and {\color{blue}$\mathbf{b_0}$} is the start event
 of transaction {\color{blue}$\mathbf{B}$}.

\begin{figure}
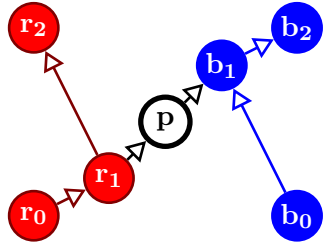

\centering
\transactionsdiagram1111111{1}
\caption{An example system state. 
The events {\color{red} $\mathbf{r_0},\mathbf{r_1},$} and 
 {\color{red}$\mathbf{r_2}$} form transaction {\color{red}$\mathbf R$}, and the
 events {\color{blue} $\mathbf{b_0},\mathbf{b_1},$} and
 {\color{blue}$\mathbf{b_2}$} form transaction {\color{blue}$\mathbf B$}.
Event $\mathbf p$ is not part of either transaction.
It may be an input, such as a network delay event, or part of a protocol used 
 to schedule the transactions.
In this state, {\color{red}$\mathbf{r_1}$}\before $\mathbf p$ \before
 {\color{blue}$\mathbf{b_1}$}, which means that {\color{red}$\mathbf{r_1}$} 
 happens before {\color{blue}$\mathbf{b_1}$}, and so the transactions are 
 ordered: {\color{red}$\mathbf R$}\before {\color{blue}$\mathbf B$}.}
\label{fig:red-blue-transactions}
\end{figure}

In order to discuss what it means to \ti{serialize} transactions, we need a
 notion of the \ti{order} in which transactions happen.
We obtain this ordering
 by lifting the happens-before relation on events to
a happens-before $\p\before$ relation for transactions.
We say that transaction $T_2$ \ti{directly depends} on $T_1$, written
$T_1 \prec T_2$, if
an event in $T_1$ happens before an event in $T_2$:
\[
T_1\prec T_2 \quad \equiv \quad 
T_1\ne T_2 \land
\exists e_1\in T_1,e_2\in T_2~.~e_1\before e_2
\]
The happens-before relation on transactions $\p\before$ is the transitive 
 closure of this direct dependence relation $\prec$.
Thus, in \cref{fig:red-blue-transactions}, the ordering {\color{red}$\mathbf 
 R$}\before {\color{blue}$\mathbf B$} holds.
Likewise, \cref{fig:hospital-events} is a system state featuring the 
 transactions from our hospital example (\cref{fig:secure-hospital}), in 
 which {\color{blue}\Patsy}\before{\color{red}\Attacker} holds.

\begin{definition}[Serializability]
\label{definition:serializability}
Transactions are serializable exactly when happens-before is a strict partial 
 order on transactions.
\end{definition}
Any total order consistent with this strict
 partial order would then respect the happens-before
 ordering $\p\before$ of events.
Such a total ordering would represent a \ti{serial order} of transactions.

\begin{definition}[\TransactionSecurity]
  \label{def:securetransaction}
  A transaction is \ti{\TransactionSecurityAdjective} if happens-before
   $\p\before$ relationships between transaction events—and therefore
   causality—are consistent with permitted information flow:
  \[
    e_1\before e_2\quad\Longrightarrow\quad\ell\p{e_1}\less\ell\p{e_2}
  \]
\end{definition}

This definition represents traditional information flow control
 within each transaction.
Intuitively, each transaction itself cannot cause a security breach 
 (although this definition says nothing about the protocol scheduling them).
In our hospital example, {\color{blue}\Patsy}'s transaction in 
 Fig.~\ref{fig:insecure-hospital-positive} is not
 \ti{\TransactionSecurityAdjective}, since
 {\color{blue}Read HIV} happens before {\color{blue}Read address}, and yet 
 the label of {\color{blue}Read HIV}, \texttt{H}, does not flow to the 
 label of {\color{blue}Read address}, \texttt{L}.
However, in the modified, secure version 
(Fig.~\ref{fig:secure-hospital-positive}), there are no such insecure 
 happens-before relationships, so {\color{blue}\Patsy}'s transaction is 
 secure.\ifreport\vspace{-3mm}\fi

\begin{figure}
\newcommand{\hospitaltwopc}[0] {{
      \begin{tikzpicture} [scale=0.9,auto=center]
    \scope
    \def\lw{1pt}
        \node[inner sep=0pt] (left) at (-2.2,7) {};
        \node[inner sep=0pt] (right) at (17,7) {};
        \node[inner sep=0pt,label={\large\tb{High Security}}] (high) at (15.5,7) {};
        \node[inner sep=0pt,yshift=2pt,label={below:\large\tb{Low Security}}] (high) at (15.5,7) {};
        \draw [black,-,dashed,line width=1pt] (left) -- (right);
      \def\labelsize{\large}
      \def\nodesize{7mm}
    \tikzstyle{every node} = [rounded corners=3mm,
                              fill=blue,
                              minimum height=6mm]
        \node (patsystart) at (-1.5,0) {\color{white} \Patsy\ start};
        \node (readhiv) at (4,9) {\color{white} Read HIV};
        \node (readaddress) at (7,6) {\color{white} Read address};
        \node (patsyprint) at (-1,20) {\color{white} Print address};

        \foreach \from/\to in {patsystart/readhiv,readhiv/readaddress,readaddress/patsyprint}
           \draw [blue,->,fill=white,line width=1pt,>=open triangle 60] (\from) -- (\to);

    \tikzstyle{every node} = [rounded corners=3mm,
                              fill={rgb:red,1;black,1},
                              minimum height=6mm]
        \node (attackerstart) at (16,0) {\color{white} \Attacker\ start};
%
    \tikzstyle{every node} = [rounded corners=2.5mm,
                              draw=black,
                              line width=2pt,
                              fill=none,
                              minimum height=5mm]
    \node (patsysend) at (-1,1.6) {\begin{tabular}{c}{\color{blue}\Patsy} sends\\message\\to database\end{tabular}};
    \node (patsyreceive) at (2,3.6) {\begin{tabular}{c}database receives\\{\color{blue}\Patsy}'s message\end{tabular}};

    \node (patsyhivlock) at (2,8) {{\color{blue}\Patsy} {\color{DarkOrchid} acquires HIV lock}};
    \node (patsyaddresslock) at (4,1) {{\color{blue}\Patsy} {\color{DarkOrchid} acquires address lock}};
    \node (patsyyes) at (7,12) {\begin{tabular}{c}database\\sends {\color{blue}\Patsy}'s\\\ti{ready} message\end{tabular}};
    \node (patsyyesreceive) at (-1,12) {\begin{tabular}{c}{\color{blue}\Patsy} receives\\\ti{ready} message\end{tabular}};
    \node (patsycommit) at (-1,14) {\begin{tabular}{c}{\color{blue}\Patsy} sends\\\ti{commit} message\end{tabular}};
    \node (patsycommitreceive) at (10,14) {\begin{tabular}{c}database \\receives {\color{blue}\Patsy}'s\\\ti{commit} message\end{tabular}};
    \node (patsyhivlockrelease) at (9,10) {{\color{blue}\Patsy} {\color{DarkOrchid} releases HIV lock}};
    \node (patsyaddresslockrelease) at (11.7,6) {{\color{blue}\Patsy} {\color{DarkOrchid} releases address lock}};

    \node (attackersend) at (16,1.5) {\begin{tabular}{c}{\color{red}\Attacker} sends\\message\\to database\end{tabular}};
    \node (attackerreceive) at (11.8,1) {\begin{tabular}{c}database receives\\{\color{red}\Attacker}'s message\end{tabular}};
    \node (attackeraddresslock) at (7.6,2) {{\color{red}\Attacker} {\color{DarkOrchid} tries to acquire address lock}};
    \node (attackerabort) at (11.5,3.3) {\begin{tabular}{c}Database sends {\color{red}\Attacker}\\an \ti{abort} message\end{tabular}};
    \node (attackerabortreceive) at (16,3.3) {\begin{tabular}{c}{\color{red}\Attacker} receives \\\ti{abort} message\end{tabular}};

%
    \foreach \from/\to in {patsystart/patsysend,patsysend/patsyreceive,patsyreceive/patsyhivlock,patsyhivlock/readhiv,readhiv/patsyaddresslock,patsyaddresslock/readaddress,readaddress/patsyyes,patsyyes/patsyyesreceive,patsyyesreceive/patsycommit,patsycommit/patsycommitreceive,patsycommit/patsyprint,patsycommitreceive/patsyhivlockrelease,patsycommitreceive/patsyaddresslockrelease,readhiv/patsyhivlockrelease,readaddress/patsyaddresslockrelease,attackerstart/attackersend,attackeraddresslock/attackerabort,attackerabort/attackerabortreceive,patsyaddresslock/attackeraddresslock,attackeraddresslock/patsyaddresslockrelease,attackersend/attackerreceive,attackerreceive/attackeraddresslock}
      \draw [black,->,line width=\lw,>=open triangle 60] (\from) -- (\to);
    \foreach \from/\to in {readhiv/patsyaddresslock,patsyaddresslock/attackeraddresslock,attackeraddresslock/attackerabort,attackerabort/attackerabortreceive}
      \draw [orange,->,line width=\lw,>=open triangle 60] (\from) -- (\to);
  \endscope
\end{tikzpicture}
}}
\centering
\hospitalfull1424
\caption{
  A possible system state after running transactions from 
   Fig.~\ref{fig:secure-hospital-positive}, assuming the patient has HIV, and 
     an exclusive lock is used to order the transactions.
  (Events prior to everything in both transactions are not shown.)
  Because {\color{blue}\Patsy} acquires the lock first, the transactions are 
   ordered {\color{blue}\Patsy}\before{\color{red}\Attacker}.
  While each transaction is \TransactionSecurityAdjective (a property of events
   within a transaction), when {\color{blue}\Patsy} releases the lock after her
   transaction, a high security event happens before a low security one.
  We discuss secure scheduling protocols in \cref{sec:stagedcommit}.
}
\label{fig:hospital-events}
\end{figure}

\subsection{Network and Timing}\ifreport\vspace{-3mm}\fi
Although this model abstracts over networks and messaging, we consider
a message to comprise both a \ti{send event} and a \ti{receive event}.
We assume asynchronous messaging: no guarantees can be made about network 
 delay.
Perhaps because this popular assumption is so daunting, many security 
 researchers ignore timing-based attacks.
There are methods for mitigating leakage via timing
 channels~\cite{Kopf:Durmuth:CSF2009,azm10,barthe2006} but in this work we too ignore
 timing.

To model nondeterministic message delay, we introduce a \ti{network delay event} for each 
 message receipt event, with the same label and location.
The network delay event may occur at any time after the message send event.
It must happen before $\p\before$ the corresponding receipt event.
In \cref{fig:red-blue-transactions}, event {\color{red}$\mathbf{r_1}$}
 could represent sending a message, event $\mathbf{p}$ could be the
 corresponding network delay event, which is not part of any transaction, and
 event {\color{blue}$\mathbf{b_1}$} could be the message receipt event.
\cref{fig:red-blue-transactions} does not require $\mathbf p$ to be a 
 network delay event.
It could be any event that is not part of either transaction.
For example, it might be part of some scheduling protocol.

\begin{figure*}[t]
\newcommand\boxlabelsize{\large}
\centering
\begin{tabular}{|r | c | c |  c |  c |  c |  c |  c |  c|}\hline
Event Scheduled:&
  & {\color{red}\boxlabelsize$\mathbf{r_0}$ }
  & {\color{red}\boxlabelsize$\mathbf{r_1}$ }
  & {\color{red}\boxlabelsize$\mathbf{r_2}$ }
  & {\color{blue}\boxlabelsize$\mathbf{b_0}$} 
  & {\color{black}\boxlabelsize$\mathbf{p}$ }
  & {\color{blue}\boxlabelsize$\mathbf{b_1}$ }
  & {\color{blue}\boxlabelsize$\mathbf{b_2}$}\\
Resulting State:&
$\cb{}$ &
\transactionsdiagram10000000 &
\transactionsdiagram11000000 &
\transactionsdiagram11100000 &
\transactionsdiagram11110000 &
\transactionsdiagram11110010 &
\transactionsdiagram11111010 &
\transactionsdiagram11111110 
\\\hline
Event Scheduled:&
  & {\color{blue}\boxlabelsize$\mathbf{b_0}$} 
  & {\color{red}\boxlabelsize$\mathbf{r_0}$ }
  & {\color{red}\boxlabelsize$\mathbf{r_1}$ }
  & {\color{black}\boxlabelsize$\mathbf{p}$ }
  & {\color{blue}\boxlabelsize$\mathbf{b_1}$ }
  & {\color{blue}\boxlabelsize$\mathbf{b_2}$}
  & {\color{red}\boxlabelsize$\mathbf{r_2}$ }\\
Resulting State:&
$\cb{}$ &
\transactionsdiagram00010000 &
\transactionsdiagram10010000 &
\transactionsdiagram11010000 &
\transactionsdiagram11010010 &
\transactionsdiagram11011010 &
\transactionsdiagram11011110 &
\transactionsdiagram11111110 
\\\hline
\end{tabular}
\caption{
Two equivalent full executions for the system state from 
 \cref{fig:red-blue-transactions}.
Each begins with a start state (the empty set for full
executions),
 followed by a sequence of events, each of which corresponds to the resulting
 system state.}
\label{fig:red-blue-executions}
\end{figure*}

\subsection{Executions, Protocols, and Inputs}
\label{sec:executionsprotocolsinputs}
An \ti{execution} is a start state paired with a totally ordered sequence of 
 events that occur after the start state.
This sequence must be consistent with happens-before $\p\before$.
Recall that a system state is a set of events (\cref{sec:state}).
Each event in the sequence therefore corresponds to a system state containing
 all the events in the start state, and all events up to and including this
 event in the sequence.
Viewing an execution as a sequence of system states, an event is
 \ti{scheduled} if it is in a state, and once it is scheduled, it will be 
 scheduled in all later states.
Two executions are \ti{equivalent} if their start states are equal, and their 
 sequences contain the same set of events, so they finish with equal system 
 states (same set of events, same \before).
A \ti{full execution} represents the entire lifetime of the system, so its start
 state contains no events.

For example, \cref{fig:red-blue-executions} illustrates two equivalent
 full executions ending in the system state from
 \cref{fig:red-blue-transactions}.

A transaction scheduling \ti{protocol} determines the order in which each 
 location schedules the events of transactions.
Given a set of possible 
 transactions, a location, and a set of events representing a system state at 
 that location, a protocol decides which event is scheduled next by the location:
\[
protocol : \textbf{set}\left<\textrm{Transactions}\right> \times \textrm{Location} \times \textrm{State} \rightarrow \textrm{event}
\]

Protocols can schedule an event from a started (but unfinished) transaction, or
 other events used by the protocol itself.
In order to schedule transaction events in ways that satisfy certain 
 constraints, like serializability, protocols may have to schedule additional 
 events, which are not part of any transaction.
These can include message send and receipt events.
For example, in \cref{fig:hospital-events}, the locking events are not 
 part of any transaction, but are scheduled by the protocol in order to ensure 
 serializability.

Certain kinds of events are not scheduled by protocols, because
they are not under the control of the system.
Events representing external inputs, including the start events of transactions, 
can happen at any time: they are fundamentally nondeterministic.
We also treat the receive times of messages as external inputs.
Each message receive event is the deterministic result of its send event and 
of a nondeterministic \ti{network delay event} featuring the same security label as the 
receive event. 
We refer to start and network delay events collectively as \ti{nondeterministic
 input events} (NIEs).


Protocols do not output NIEs.
Instead, an NIE may appear at any point in an execution, and any prior events in 
 the execution can happen before $\p\before$ the NIE.
Recall that an execution features a sequence of events, each of which can be
 seen as a system state featuring all events up to that point.
An execution $E$ is consistent with a protocol $p$ if every event in the
 sequence is either an NIE, or the result of $p$ applied to the previous state
 at the event's location.
We sometimes say $p$ \ti{results in} $E$ to mean ``$E$ is consistent with $p$.''

As an example, assume all events in \cref{fig:red-blue-transactions} 
 have the same location $L$, and no messages are involved.
Start events {\color{red} $\mathbf{r_0}$} and {\color{blue} $\mathbf{b_0}$} 
 are NIEs.
Every other event has been scheduled by a protocol.
\cref{fig:red-blue-executions} shows two different executions, which may 
 be using different protocols, determining which events to schedule in each 
 state.
We can see that in the top execution of \cref{fig:red-blue-executions}, 
 the protocol maps:
\[
\begin{array}{rl}
  \cb{{\color{red} \mathbf{R}},{\color{blue} \mathbf{B}},\dots} , L, 
  \cb{{\color{red} \mathbf{r_0}}} \mapsto&\!\! 
  {\color{red} \mathbf{r_1}}
  \\
  \cb{{\color{red} \mathbf{R}},{\color{blue} \mathbf{B}},\dots} , L, 
  \cb{{\color{red} \mathbf{r_0}},{\color{red} \mathbf{r_1}}} \mapsto&\!\! 
  {\color{red} \mathbf{r_2}}
  \\
  \cb{{\color{red} \mathbf{R}},{\color{blue} \mathbf{B}},\dots} , L, 
  \cb{{\color{red} \mathbf{r_0}},{\color{red} \mathbf{r_1}},{\color{red} \mathbf{r_2}},
      {\color{blue} \mathbf{b_0}}} \mapsto&\!\!  
  {\color{black} \mathbf{p}}
  \\
  \cb{{\color{red} \mathbf{R}},{\color{blue} \mathbf{B}},\dots} , L, 
  \cb{{\color{red} \mathbf{r_0}},{\color{red} \mathbf{r_1}},{\color{red} \mathbf{r_2}},
      {\color{blue} \mathbf{b_0}},{\color{black} \mathbf{p}}} \mapsto&\!\! 
  {\color{blue} \mathbf{b_1}}
  \\
  \cb{{\color{red} \mathbf{R}},{\color{blue} \mathbf{B}},\dots} , L, 
  \cb{{\color{red} \mathbf{r_0}},{\color{red} \mathbf{r_1}},{\color{red} \mathbf{r_2}},
      {\color{blue} \mathbf{b_0}},{\color{black} \mathbf{p}},
      {\color{blue} \mathbf{b_1}}} \mapsto&\!\! 
  {\color{blue} \mathbf{b_2}}
\end{array}
\]
The protocol in the bottom execution of \cref{fig:red-blue-executions}
 maps:
\[
\begin{array}{rl}
  \cb{{\color{red} \mathbf{R}},{\color{blue} \mathbf{B}},\dots} , L, 
  \cb{{\color{red} \mathbf{r_0}},{\color{blue} \mathbf{b_0}}} \mapsto&\!\! 
  {\color{red} \mathbf{r_1}}
  \\
  \cb{{\color{red} \mathbf{R}},{\color{blue} \mathbf{B}},\dots} , L, 
  \cb{{\color{red} \mathbf{r_0}},{\color{blue} \mathbf{b_0}},
      {\color{red} \mathbf{r_1}}} \mapsto&\!\! 
  {\color{black} \mathbf{p}}
  \\
  \cb{{\color{red} \mathbf{R}},{\color{blue} \mathbf{B}},\dots} , L, 
  \cb{{\color{red} \mathbf{r_0}},{\color{blue} \mathbf{b_0}},
      {\color{red} \mathbf{r_1}},{\color{black} \mathbf{p}}} \mapsto&\!\! 
  {\color{blue} \mathbf{b_1}}
  \\
  \cb{{\color{red} \mathbf{R}},{\color{blue} \mathbf{B}},\dots} , L, 
  \cb{{\color{red} \mathbf{r_0}},{\color{blue} \mathbf{b_0}},
      {\color{red} \mathbf{r_1}},{\color{black} \mathbf{p}},
      {\color{blue} \mathbf{b_1}}} \mapsto&\!\! 
  {\color{blue} \mathbf{b_2}}
  \\
  \cb{{\color{red} \mathbf{R}},{\color{blue} \mathbf{B}},\dots} , L, 
  \cb{{\color{red} \mathbf{r_0}},{\color{blue} \mathbf{b_0}},
      {\color{red} \mathbf{r_1}},{\color{black} \mathbf{p}},
      {\color{blue} \mathbf{b_1}},{\color{blue} \mathbf{b_2}}} \mapsto&\!\! 
  {\color{red} \mathbf{r_2}}
\end{array}
\]

Ultimately, a protocol must determine the ordering of transactions.
If the exact set of start events to be scheduled (as opposed to start events
 possible) were always known in advance, scheduling would be trivial.
A protocol should not require one transaction to run before another \ti{a
 priori}:
\tb{start events from any subset of possible transactions may be scheduled at 
 any time.}
No protocol should result in a system state in which such a start event 
 cannot be scheduled, or an incomplete transaction can never finish.

\subsection{Semantic Security Properties}
\label{sec:semanticsecurity}
Consider an observer who can only ``see'' events at some security level $\ell$ 
 or below.
If two states $S_1$ and $S_2$ are indistinguishable to the observer, then 
 after a program runs, noninterference requires that the resulting 
 executions remain indistinguishable to the observer.
Secret values, which the observer cannot see, may differ in $S_1$ and $S_2$, 
 and may result in different states at the end of the executions, but the 
 observer should not be able to see these differences.

\ifreport
\subsubsection{Possibilistic Noninterference}
\label{sec:possibilistic}
David Sutherland's hyperproperty \ti{Generalized 
 Noninterference}\footnote{McCullough coins the term ``Generalized 
 Noninterference''~\cite{mccu88}, and Clarkson and Schneider define 
 hyperproperties~\cite{cs08}.}~\cite{Suther86} generalizes Goguen and 
 Meseguer's noninterference~\cite{GM82}.
His model features ``possible execution sequences'', much like our 
 \ti{executions}, each of which is a  sequence of system states.
For a given observer, some information is \ti{low observable,} meaning the 
 observer may learn it.
Other information is \ti{high}, meaning it's too secret for the observer to 
 know.
His model also features some events, called ``signals,'' representing 
 \ti{inputs}, which can be either low or high.
\ICS{Not sure if I'm staying in Sutherland's model too long, or if I should 
 basically drop his name and move on to explaining it in our model.}
Possibilistic Noninterference, then, requires that for any given execution 
 $E_1$, it must be possible to change the high inputs of $E_1$ to those of any 
 other valid execution $E_2$, and create a valid, \ti{possible} execution 
 $E_3$ without changing any low events:
\ICS{This sentence is terrible.}
\[
\forall E_1,E_2. \exists E_3. 
\ \ \begin{array}{r l r l}
  High\_inputs&\!\!\!\!\!\p{E_3} = &\!\!\!\!High\_inputs&\!\!\!\!\!\p{E_2} \land \\
  Low\_events&\!\!\!\!\!\p{E_3} = &\!\!\!\!Low\_events&\!\!\!\!\!\p{E_1}
    \end{array}
\]
In a sense, an observer can't make any observations that change the 
 \ti{possible} set of high inputs, but might be able to infer which are 
 probable.
This is recognized as a fairly weak form of noninterference in 
 nondeterministic systems.~\cite{cs08}

In our hospital example, as illustrated in \cref{fig:secure-hospital}, 
 the system determines which of {\color{blue}\Patsy}'s transactions will run 
 based upon whether \texttt{p.hasHiv} is \texttt{true}. We can
 treat this condition to be
 a high-security event that happens before all reads of \texttt{p.hasHiv}.
If we classify this past high-security event as input, and all low-security
 events as low-observable for {\color{red}\Attacker}, then we must ensure that 
 when {\color{blue}\Patsy}'s code runs, the set of possible low-security 
 events that result is the same regardless of whether \texttt{p.hasHiv}.
{\color{blue}\Patsy}'s possible transactions in 
 \cref{fig:secure-hospital} ensure possibilistic noninterference, while 
 her transactions in \cref{fig:insecure-hospital} do not, since whether 
 or not {\color{blue}Read address} occurs depends on \texttt{p.hasHiv}.
\fi

\subsubsection{\RelaxedSecurity}
\label{sec:limited}

Semantic conditions for information security are typically based on
some variant of noninterference~\cite{GM82,sm-jsac}. These variants
are often distinguished by their approaches to nondeterminism.
However, many of these semantic security conditions fail under
\ti{refinement}: if some nondeterministic choices are fixed, security
is violated~\cite{zm03}. However, low-security observational determinism~\cite{Roscoe95,zm03} is a
strong property that is secure under refinement:
intuitively, if an observer with label $\ell$ cannot distinguish states $S$ 
 and $S^\prime$, that observer must not be able to distinguish any 
 execution $E$ beginning with $S$ from \ti{any} execution $E^\prime$ 
 beginning with $S^\prime$:
\[
  \p{S\approx_\ell S^\prime}
  \Rightarrow 
  E\approx_\ell E^\prime
\]
This property is \emph{too} strong because it rules out two sources of nondeterminism that
we want to allow: first, the ability of any transaction to
start at any time, and second, network delays.
Therefore, we relax observational determinism to
permit certain nondeterminism.
We only require that executions be indistinguishable to the observer if
their NIEs are indistinguishable to the observer:
\[
  \p{
      S\approx_\ell S^\prime
          \land    
      \NIE\p E \approx_\ell \NIE \p{E^\prime}  
  } 
  \Rightarrow E\approx_\ell E^\prime
\]

We call this relaxed property \ti{\relaxedsecurity}.
It might appear to be equivalent to observational determinism,
but with the NIEs encoded in the start states.
This is not the case.
If NIEs were encoded in the start states, protocols would be able to read which
 transactions will start and when messages will arrive in the future.
Therefore \relaxedsecurity captures something that
 observational determinism does not: unknowable but ``allowed'' nondeterminism
 at any point in an execution.

By deliberately classifying start events and network delays as input, we allow
certain kinds of information leaks that observational determinism would not.
Specifically, a malicious network could leak information by
manipulating the order or timing of message delivery.
However, such a network could by definition communicate information to its
 co-conspirators anyway.
Information can also be leaked through the order or timing of start
events.
This problem is beyond the scope of this work.

Conditioning the premise of the security condition on the
indistinguishability of information that is allowed to be released is
an idea that has been used earlier~\cite{sm04}, but not
in this way, to our knowledge.

In our hospital example, as illustrated in \cref{fig:secure-hospital}, 
 the system determines which of {\color{blue}\Patsy}'s transactions (the one
 with the dashed events, or the one without the dashed events) will run 
 based on whether \texttt{p.hasHiv} is \texttt{true}.
We can consider \texttt{p.hasHiv}'s value to be a high-security event that
 happens before all reads of \texttt{p.hasHiv}.
If we classify this past high-security event as input, and all low-security
 events as low-observable for {\color{red}\Attacker}, then we must ensure that 
 when {\color{blue}\Patsy}'s code runs, the low-security  projections of 
 resulting executions are always the same, regardless of whether 
 \texttt{p.hasHiv}.
{\color{blue}\Patsy}'s possible transactions in \cref{fig:secure-hospital} 
 allow for observational determinism, while her transactions in 
 \cref{fig:insecure-hospital} do not, since whether or not {\color{blue}Read 
 address} occurs depends on \texttt{p.hasHiv}.
Whether or not the system actually maintains observational determinism, 
 however, depends on the protocol scheduling the events.
\ICS{This is pretty much a repeat of the version for possibilistic 
 noninterference, but since that's in the tech report, I copied the example 
 here.}

\begin{definition}[\ProtocolSecurity]
  \label{def:secureprotocol}
  A protocol is considered \ti{\ProtocolSecurityAdjective} if the set of resulting
   executions satisfies \relaxedsecurity for any allowed
   sets of \TransactionSecurityAdjective transactions and any possible NIEs.
\end{definition}


\section{Impossibility}
\label{sec:impossibility}
\label{sec:introimpossible}
One of our contributions is to show that even in the absence
of timing channels, there is a fundamental conflict between secure
noninterference and serializability. Previous results showing such a
conflict, for example the work of Smith et al.~\cite{Smith1996}
consider only confidentiality and show only that timing channels are
unavoidable.

\begin{theorem}[Impossibility]
No \ProtocolSecurityAdjective protocol\footnote{barring unforeseen cryptographic capabilities 
 (\cref{sec:crypto})} can serialize all possible sets of \TransactionSecurityAdjective 
 transactions.\footnote{In fact, what we prove is stronger.
Our proof holds for even
 possibilistic security conditions~\cite{mccu88}, which are
 weaker than \relaxedsecurity (see \supplemental).
 No protocol whose resulting traces satisfy even this weaker condition
 can serialize all sets of \TransactionSecurityAdjective transactions.}
\label{theorem:impossibility}
\end{theorem}
We assume protocols cannot simply introduce an arbitrarily trusted
third party; a protocol must be able to run using only the set of
locations that have events being scheduled.
\ifreport
\begin{proof}
(by counterexample)
\else

\begin{proof}[sketch]
\fi
\setlength{\belowcaptionskip}{-3mm} 
\begin{figure}[t]
\centering
 \begin{tikzpicture} [scale=1,auto=center]
   \def\labelsize{\large}
   \def\nodesize{7mm}
   \node[rounded corners=2mm, draw, fill=blue!10, minimum height=29mm,minimum width=20mm,text width=20mm,text height=25mm,label=Carol] (carol) at (6,3) 
     {\ \hfill \includegraphics[height=25mm]{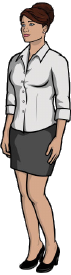}};
   \node[rounded corners=2mm, draw, fill=red!10, minimum height=29mm,minimum width=20mm,text width=20mm,text height=20mm,label=Dave] (bank) at (0,3) 
   {\hspace{-1.5mm}\includegraphics[height=20mm]{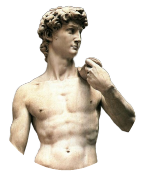}};
   \node[inner sep=0pt] (alicecloud) at (3,5.5) 
   {\includegraphics[height=20mm]{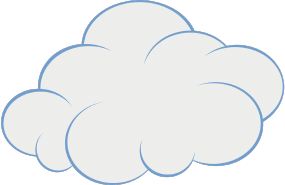}};
   \node[inner sep=0pt,label={right:Alice}] (alice) at (4.7,5.6) 
   {\includegraphics[height=20mm]{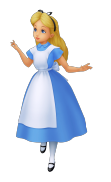}};
   \node[inner sep=0pt] (bobcloud) at (3,.5) 
   {\includegraphics[height=20mm]{icons/Cloud.pdf}};
   \node[inner sep=0pt,label={left:Bob}] (bob) at (1.2,0.3) 
   {\includegraphics[height=17mm]{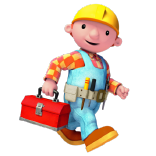}};

  \tikzstyle{every node} = [circle,
      draw={rgb:red,1;black,1},
      inner sep=0pt,
      line width=1pt,
      fill=red,
      minimum size=\nodesize]
  \node (r0) at (.6,3.9) {\color{white}\labelsize$\mathbf{r_0}$};
  \node (r1) at (.6,2.1) {\color{white}\labelsize$\mathbf{r_1}$};
  \node (r2) at (2.3,5.3) {\color{white}\labelsize$\mathbf{r_2}$};
  \node (r3) at (2.3,.5) {\color{white}\labelsize$\mathbf{r_3}$};
  \foreach \from/\to in {r0/r2,r1/r3}
      \draw [red,fill=white,->,line width=1pt,>=open triangle 60] (\from) -- (\to);

  \tikzstyle{every node} = [circle,
      inner sep=0pt,
      draw=blue,
      line width=1pt,
      fill=blue,
      minimum size=\nodesize]
  \node (b0) at (5.4,3.9) {\color{white}\labelsize$\mathbf{b_0}$};
  \node (b1) at (5.4,2.1) {\color{white}\labelsize$\mathbf{b_1}$};
  \node (b2) at (3.7,5.3) {\color{white}\labelsize$\mathbf{b_2}$};
  \node (b3) at (3.7,.5) {\color{white}\labelsize$\mathbf{b_3}$};
  \foreach \from/\to in {b0/b2,b1/b3}
      \draw [blue,fill=white,->,line width=1pt,>=open triangle 60] (\from) -- (\to);

  \node[inner sep=0pt,fill=none,draw=none] (wall) at (3,3) 
   {\includegraphics[height=30mm]{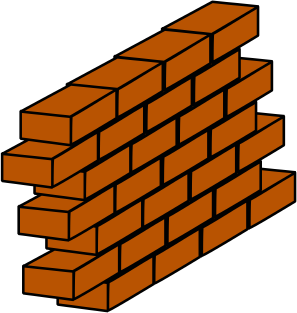}};

   \node[inner sep=0pt,minimum size=0pt,label={\LARGE\tb?}] (orderbob) at (3,.5) {};
   \node[inner sep=0pt,minimum size=0pt] (orderbob0) at (3.1,.5) {};
   \node[inner sep=0pt,minimum size=0pt] (orderbob2) at (2.9,.5) {};
   \node[inner sep=0pt,minimum size=0pt,label={\LARGE\tb?}] (order) at (3,5.3) {};
   \node[inner sep=0pt,minimum size=0pt] (order0) at (3.1,5.3) {};
   \node[inner sep=0pt,minimum size=0pt] (order2) at (2.9,5.3) {};
   \foreach \from/\to in {order0/r2,order2/b2,orderbob0/r3,orderbob2/b3}
      \draw [black,->,fill=white,line width=1pt,>=open triangle 60] (\from) -- (\to);
\end{tikzpicture}
\caption{
  Transactions that cannot be securely serialized.
  Dave's transaction includes {\color{red}$\mathbf{r_0}$}, 
   {\color{red}$\mathbf{r_1}$}, {\color{red}$\mathbf{r_2}$},  and
   {\color{red}$\mathbf{r_3}$}.
  Carol's includes
   {\color{blue}$\mathbf{b_0}$}, {\color{blue}$\mathbf{b_1}$}, 
   {\color{blue}$\mathbf{b_2}$}, and {\color{blue}$\mathbf{b_3}$}.
  Cloud providers Alice and Bob must decide how to order their events.
  Alice and Bob may not influence each other, and Carol and Dave may not
   influence each other, as represented by the wall.
  For these transactions to be serializable, Alice's ordering of
   {\color{red}$\mathbf{r_2}$} and  {\color{blue}$\mathbf{b_2}$} must agree with 
   Bob's ordering of {\color{red}$\mathbf{r_3}$} and  
   {\color{blue}$\mathbf{b_3}$}.
}
\label{fig:impossible}
\end{figure}
\setlength{\belowcaptionskip}{\correctbelowcaptionskip}
Consider the counterexample shown in \cref{fig:impossible}.
Alice and Bob are both cloud computing providers who keep strict logs of the 
 order in which various jobs start and stop.
Highly trusted (possibly government) auditors may review these logs, and check 
 for consistency, to ensure cloud providers are honest and fair.
As competitors, Alice and Bob do not want each other to gain any information 
 about their services, and do not trust each other to affect their own 
 services.

Carol and Dave are presently running jobs on Alice's cloud.
Both Carol and Dave would like to stop their jobs on Alice's cloud, and start 
 new ones on Bob's cloud.
Each wants to do this atomically, effectively maintaining exactly one running 
 job at all times.
Carol and Dave consider their jobs to be somewhat confidential; they do not 
 want each other to know about them.
Unlike the example from \cref{fig:rainforest}, Dave and Carol's  
 transactions do not go through a third party like Rainforest.
For the transactions to be serializable, Alice's ordering of the old jobs 
 stopping must agree with Bob's ordering of the new jobs starting.

\ifreport
These transactions feature at least 8 events:

\begin{tabular}{r l}
  {\color{red}$\mathbf{r_0}$}:& Dave sends a message to Alice
\\{\color{red}$\mathbf{r_1}$}:& Dave sends a message to Bob
\\{\color{red}$\mathbf{r_2}$}:& Alice receives a message from Dave, ending a job.
\\{\color{red}$\mathbf{r_3}$}:& Bob receives a message from Dave, beginning a job.
\\{\color{blue}$\mathbf{b_0}$}:& Carol sends a message to Alice
\\{\color{blue}$\mathbf{b_1}$}:& Carol sends a message to Bob
\\{\color{blue}$\mathbf{b_2}$}:& Alice receives a message from Carol, ending a job.
\\{\color{blue}$\mathbf{b_3}$}:& Bob receives a message from Carol, beginning a job.
\end{tabular}

No events at Alice's location should influence events at Bob's location, and 
 vice-versa.
No events at Carol's location should influence events at Dave's location, and 
 vice-versa. 

Alice and Bob must each finish with ordered logs including job beginnings and
 endings.
This means they must assign a happens-before $\p\before$ relation to their 
 events above.
For these transactions to be serializable, Alice's ordering of
 {\color{red}$\mathbf{r_2}$} and  {\color{blue}$\mathbf{b_2}$} must agree with 
 Bob's ordering of {\color{red}$\mathbf{r_3}$} and  
 {\color{blue}$\mathbf{b_3}$}.

\begin{lemma}{These transactions are \TransactionSecurityAdjective.}\\
  The two transactions in \cref{fig:impossible} are \TransactionSecurityAdjective
  (\cref{def:securetransaction}).
\end{lemma}
\begin{proof}
The only happens-before relationships within transactions are for the sending 
 and receipt of messages, explicitly carrying information readable to the 
 recipient.
All four are consistent with permitted information flows.
\end{proof}

\begin{lemma}{}
  No protocol can securely serialize these transactions.
  Specifically, no protocol accepting these transactions can preserve 
   possibilistic noninterference.
\end{lemma}
\begin{proof}
\fi
In any system with an asynchronous network, it is possible to reach a   state 
 in which Carol's message to Alice has arrived, but not her message to Bob, 
 and Dave's message to Bob has arrived, but not his message to Alice.
\ifreport
In other words, events {\color{red}$\mathbf{r_2}$} and 
 {\color{blue}$\mathbf{b_3}$} have not yet occurred.
\cref{fig:impossiblehalf} illustrates this situation.
\fi
In this state, neither Alice nor Bob can know whether one or both transactions 
 have begun. 
It is impossible for either to communicate this information to the other without violating 
\ifreport
 possibilistic noninterference. 
\else
 \relaxedsecurity.
\fi
Specifically, any protocol that relayed such information from one cloud provider to the 
 other would allow the recipient to distinguish the order of message delivery to the 
 other cloud provider.
That ordering is considered secret input,
and so this would be a security violation.
All executions with identical start states, and identical inputs visible to 
 Alice, but differently ordered network delay events at Bob, which are inputs 
 invisible to Alice, would become distinguishable to Alice.
\ifreport
Even possibilistic noninterference would therefore be violated 
 (\cref{sec:possibilistic}).

\begin{figure}[t]
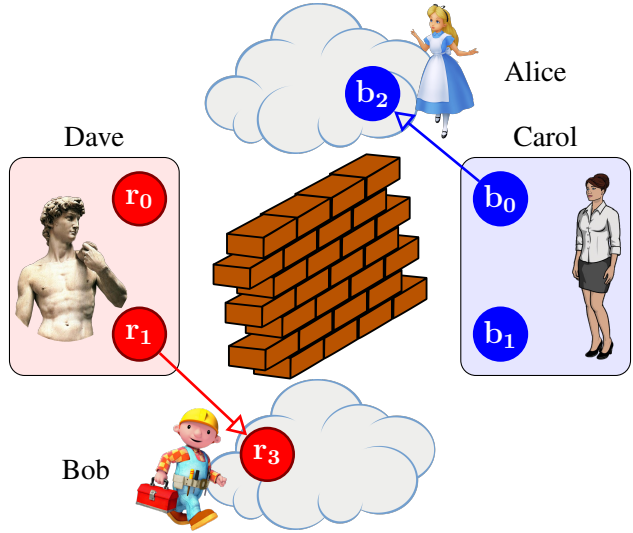

\centering
 \begin{tikzpicture} [scale=1,auto=center]
   \def\labelsize{\large}
   \def\nodesize{7mm}
   \node[rounded corners=2mm, draw, fill=blue!10, minimum height=29mm,minimum width=20mm,text width=20mm,text height=25mm,label=Carol] (carol) at (6,3) 
     {\ \hfill \includegraphics[height=25mm]{icons/Carol.pdf}};
   \node[rounded corners=2mm, draw, fill=red!10, minimum height=29mm,minimum width=20mm,text width=20mm,text height=20mm,label=Dave] (bank) at (0,3) 
   {\hspace{-1.5mm}\includegraphics[height=20mm]{icons/Dave.pdf}};
   \node[inner sep=0pt] (alicecloud) at (3,5.5) 
   {\includegraphics[height=20mm]{icons/Cloud.pdf}};
   \node[inner sep=0pt,label={right:Alice}] (alice) at (4.7,5.6) 
   {\includegraphics[height=20mm]{icons/Alice.pdf}};
   \node[inner sep=0pt] (bobcloud) at (3,.5) 
   {\includegraphics[height=20mm]{icons/Cloud.pdf}};
   \node[inner sep=0pt,label={left:Bob}] (bob) at (1.2,0.3) 
   {\includegraphics[height=17mm]{icons/Bob.pdf}};

  \tikzstyle{every node} = [circle,
      draw={rgb:red,1;black,1},
      inner sep=0pt,
      line width=1pt,
      fill=red,
      minimum size=\nodesize]
  \node (r0) at (.6,3.9) {\color{white}\labelsize$\mathbf{r_0}$};
  \node (r1) at (.6,2.1) {\color{white}\labelsize$\mathbf{r_1}$};
  \node (r3) at (2.3,.5) {\color{white}\labelsize$\mathbf{r_3}$};
  \foreach \from/\to in {r1/r3}
      \draw [red,fill=white,->,line width=1pt,>=open triangle 60] (\from) -- (\to);

  \tikzstyle{every node} = [circle,
      inner sep=0pt,
      draw=blue,
      line width=1pt,
      fill=blue,
      minimum size=\nodesize]
  \node (b0) at (5.4,3.9) {\color{white}\labelsize$\mathbf{b_0}$};
  \node (b1) at (5.4,2.1) {\color{white}\labelsize$\mathbf{b_1}$};
  \node (b2) at (3.7,5.3) {\color{white}\labelsize$\mathbf{b_2}$};
  \foreach \from/\to in {b0/b2}
      \draw [blue,fill=white,->,line width=1pt,>=open triangle 60] (\from) -- (\to);

  \node[inner sep=0pt,fill=none,draw=none] (wall) at (3,3) 
   {\includegraphics[height=30mm]{icons/Wall.pdf}};

\end{tikzpicture}
\caption{An intermediate state of an execution featuring the transactions from 
   \cref{fig:impossible}.}
\label{fig:impossiblehalf}
\end{figure}

Additionally, we have assumed that a protocol must be able to schedule any 
 subset of the allowed transactions' start events.
Therefore valid executions exist in which, say, only Carol's transaction runs, 
 so Alice receives only information about Carol's transaction, and commits 
 Carol's transaction first.
Therefore a valid execution must exist in which Alice commits Carol's 
 transaction first, before receiving any further input from Dave or Bob, and 
 likewise, Bob commits Dave's transaction first, without further input from 
 Carol or Alice. 
Thus any protocol satisfying possibilistic noninterference can schedule 
 inconsistently: the transactions cannot be securely serialized.
\end{proof}

Thus, with this scenario as a counterexample, no \ProtocolSecurityAdjective protocol can serialize 
 all possible sets of \TransactionSecurityAdjective transactions. 
\end{proof}
\else
\ICS{Well, we've really crunched down this not-tech-report proof sketch.}
\end{proof}
\fi

\subsection{Cryptography}
\label{sec:crypto}
This essentially information-theoretic argument does not account for the possibility that some protocol could 
 produce \ti{computationally indistinguishable} traces that are
 low-distinguishable with sufficient computational power (e.g., to break encryption).
However, we are unaware of any cryptographic protocols that would permit 
 Alice and Bob to learn a consistent order in which to schedule events
 without learning each other's confidential information.

\section{Analysis}
\label{sec:analysis}
Although secure scheduling is impossible in general, many sets of transactions can 
 be scheduled securely.
We therefore investigate which conditions are sufficient for secure 
 scheduling, and what protocols can function securely under these conditions.

\subsection{Monotonicity}
A relatively simple condition suffices to guarantee schedulability, while 
 preserving \relaxedsecurity:
\begin{definition}[Monotonicity]
  A transaction is \ti{monotonic} if it is \TransactionSecurityAdjective
  and its events are totally ordered by
  happens-before $\p\before$.
\label{definition:monotonicity}
\end{definition}
\begin{theorem}[Monotonicity $\Rightarrow$ Schedulability]\ \\
\label{theorem:monotonicity}
A protocol exists that can serialize any set of monotonic transactions 
 and preserve \relaxedsecurity.
\end{theorem}
\ifreport
\begin{proof}
\else
\begin{proof}[sketch]
\fi
Monotonicity requires that each event must be allowed to influence all future 
 events in the transaction.
A simple, pessimistic transaction protocol can schedule such transactions 
 securely.
In order to define this protocol, we need a notion of \ti{locks} within our 
 model.

\noinunbo{Locks}
A lock consists of an infinite set of events for each allowed transaction.
A transaction \ti{acquires} a lock by scheduling any event from this set.
It \ti{releases} a lock by scheduling another event from this set.
Thus, in a system state $S$, a transaction $T$ \ti{holds} a lock if $S$ 
 contains an odd number of events from the lock's set corresponding to $T$.
No correct protocol should result in a state in which multiple transactions 
 hold the same lock.
All pairs of events in a lock conflict, so scheduled events that are part
 of the same lock must be totally ordered by happens-before $\p\before$.
All events in a lock share a location, which is considered to be the
 location of the lock itself.
Likewise, all events in a lock share a label, which is considered to be the
 label of the lock itself.

A critical property for transaction scheduling is
 \ti{deadlock freedom}~\cite{2phase,dinosaurbook},
which requires that a protocol can eventually schedule all events 
 from any transaction whose start event has been scheduled.
A system enters \ti{deadlock} when it reaches a state after which this is not 
 the case.
For example, deadlock happens if a protocol requires two transactions each to 
 wait until the other completes: both will wait forever.
%
If all transactions are finite sets of events (i.e., all 
 transactions can terminate), then deadlock freedom guarantees that a system 
 with a finite set of start events eventually terminates, a liveness property.
\ifreport
\ICS{RVR finds this sentence unnecessary, so I shunted it to the TR.}
Deadlock freedom is essential to distributed or parallel scheduling, but 
 notoriously difficult to get right~\cite{dinosaurbook}.
\fi

We now describe a deadlock-free protocol that can securely serialize
any set of monotonic
 transactions, and preserve \relaxedsecurity:

\begin{itemize}
\item Each event in each transaction has a corresponding lock, except start events.
\item Any events that have the same label share a lock, and this
  lock shares a location with at least one of the events.
  Conflicting events are assumed to share a label~(\cref{sec:serializability}).
\item A transaction must hold an event's lock to schedule that event.
\item A transaction acquires locks in sequence, scheduling events as it goes.
      Since all events are ordered according to a global security lattice, all
       transactions that acquire the same locks do so in the same order.
      Therefore they do not deadlock.
\item If a lock is already held, the transaction waits for it to be released.
\item When all events are scheduled, the transaction commits, releasing locks
   in reverse order. 
  Any messages sent as part of the transaction would thus receive a
  reply, indicating only that the message had been received,
   and all its repercussions committed.
  We call these replies \ti{commit} messages.
\item For each location, the protocol rotates between all uncommitted 
   transactions, scheduling any intermediate events (such as lock acquisitions) 
   until it either can schedule one event in the transaction or can make no 
   progress, and then rotates to the next transaction.
\end{itemize}
\noinunbo{Security Intuition}
Acquiring locks shared by multiple events on different locations requires a  
 commit protocol between those locations.
However, this does not leak information because all locations involved are 
 explicitly allowed to observe and influence all events involved.
Therefore several known commit protocols will do, including 2PC.
Since the only messages sent as part of the protocol are commit messages, and
 each recipient knows it will receive a commit message by virtue of sending a
 message in the protocol, no information (other than timing) is transferred by
 the scheduling mechanism itself.
\ifreport

\noinunbo{\Relaxedsecurity}
This protocol, implemented with monotonic transactions, satisfies \relaxedsecurity,
 our slightly relaxed version of observational determinism (\cref{sec:limited}).
We consider an event \ti{observable} to an observer with label $\ell$ if the
 label of the event flows to $\ell$.
For any two executions beginning with equivalent states (for some observer 
 $\ell$),
\[
 E_0\sqb0 \approx_\ell E_1\sqb0 
\]
If the executions $E_0$ and $E_1$ have the same $\ell$-observable inputs, which
 is to say transaction start events and network delay events, then the 
 protocol requires $E_0$ and $E_1$ to be indistinguishable to $\ell$.
The observer of label $\ell$ can only observe a prefix of each transaction being
 scheduled in a round-robin fashion, and commit messages for each arriving
 sometime thereafter.
Arrival time of these commit messages is considered an input, and so all 
 events visible in $E_0$ and $E_1$ are deterministic results of the events 
 visible in the start states, and the NIEs.
Each distinct state in an execution, as observed at $\ell$, will be 
 deterministically predicted by prior states and inputs.
Thus \relaxedsecurity is preserved.

\noinunbo{Serializability}
Transactions consist of totally ordered series of events.
Let $e_1$ be the first event in $T_1$ conflicting with any event in $T_2$.
Let $e_2$ be the event in $T_2$ with which $e_1$ conflicts. 
Suppose they are scheduled such that $e_1\before e_2$.
Therefore all events in $T_2$ after and including $e_2$ cannot be scheduled 
 until $T_1$ commits and releases its locks.
No event in $T_2$ scheduled before $e_2$ can conflict with an event in $T_1$ 
 after $e_1$, by monotonicity, or before $e_1$, by the definition of $e_1$.
Thus all conflicting events in $T_2$ are scheduled after all events in $T_1$, 
 so no event in $T_1$ can happen after an event in $T_2$.
Therefore, this pessimistic protocol ensures serializability.

\noinunbo{Liveness}
This scheduling system cannot result in deadlock, since all transactions 
 acquire locks in strictly increasing order on the lattice, so any set of 
 transactions that acquire the same locks must do so in the same order.
\ICS{Does this need to be further explained? If you're used to reasoning about 
 this sort of thing, it's pretty obvious... }
\ACM{This seems okay but I think it belongs in the main paper.}

Therefore, monotonicity is sufficient to guarantee secure schedulability.
\fi
\end{proof}

\subsection{Relaxed Monotonicity}
Monotonicity, while relatively easy to understand, is not the weakest 
 condition we know to be sufficient for secure schedulability.
It can be substantially relaxed.
In order to explain our weaker condition, \ti{relaxed monotonicity}, we first 
 need to introduce a concept we call \ti{visibility}:

\begin{definition}[Visible-To]
\label{definition:visible}
An event $e$ in transaction $T$ is \emph{visible to} a location $L$
 if and only if it happens at $L$, or if there exists another event $e^\prime\in 
 T$ at $L$, such that $e\before e^\prime$.
\end{definition}

\begin{definition}[Relaxed Monotonicity]
\label{definition:relaxedmonotonicity}
A transaction $T$ satisfies \ti{relaxed monotonicity} if it is \TransactionSecurityAdjective
and for each
location $L$, all events in $T$ visible to
 $L$ happen before all events in $T$ not visible to $L$.
\end{definition}

\JED{Would be nice to have some discussion of relaxed monotonicity
here, perhaps giving some intuition for why it is sufficient for
security.}
\TMreply{While I agree, I think that it'd come off as repetitive being this
close to the section where it's discussed.}
In \cref{sec:stagedcommit}, we demonstrate that relaxed monotonicity 
 guarantees schedulability.
Specifically, we present a staged commit protocol, and prove that it
 schedules any set of transactions satisfying relaxed monotonicity, while 
 preserving \relaxedsecurity 
 (\cref{theorem:stagedcommit}).

\subsection{Requirements for Secure Atomicity}
Monotonicity and relaxed monotonicity are sufficient conditions for a set of 
 transactions to be securely schedulable.
Some sets of transactions meet neither condition, but can be 
 securely serialized by some protocol.
For example, any set of transactions that each happen entirely at one 
 location can be securely serialized if each location schedules each 
 transaction completely before beginning the next.
We now describe a relatively simple condition that is necessary for any set 
 of transactions to be securely scheduled.

\subsubsection*{Decision Events and Conflicting Events}
In order to understand this necessary condition, we first describe 
 \ti{decision events} and \ti{conflicting events}.

Borrowing some terminology from Fischer, Lynch, and 
 Paterson~\cite{Fischer82b}, for a pair of transactions $T_1$ and $T_2$, any 
 system state is either \ti{bivalent} or \ti{univalent}.
A system state is \ti{bivalent} with respect to $T_1$ and $T_2$ if
there exist two valid executions that both 
 include that state, but end with opposite orderings of $T_1$ and $T_2$.
A system state is \ti{univalent} with respect to $T_1$ and $T_2$
otherwise: for one ordering of the 
 transactions, no valid execution ending with that ordering contains the 
 state.

We can define a similar relationship for start events:
for any pair of distinct start events $s_1$ and $s_2$, a system state is
 bivalent with respect to those events if it features in two valid
 executions, both of which have $s_1$ and
 $s_2$ in scheduled transactions, but those transactions are in opposite order.
A system state is univalent with respect to $s_1$ and $s_2$ otherwise.

All full executions (i.e., those starting with an empty state) that order a 
 pair of transactions begin in a \ti{bivalent} state with respect to their start
 events, before either is scheduled.
By our definition of serializability and transaction ordering, once
 transactions are ordered, they cannot be un-ordered.
Any execution that orders the transactions therefore ends in a \ti{univalent} 
 state with respect to their start events.
Any such execution consists of a sequence of 0 or more \ti{bivalent} 
 states followed by a sequence of \ti{univalent} states.
The event that is scheduled in the first \ti{univalent} state, in a sense, 
 decides the ordering of the transactions. 
We call it the \ti{decision event}.

We call any event in $T_1$ or $T_2$ that conflicts with an event in the other 
 transaction a \ti{conflicting event}.

\begin{lemma}[Decision Event \before\ Conflicting Events]\ \\
For any univalent state $S$ with $T_1\before T_2$, there exists a full
 execution $E$ ending in $S$ featuring a decision event $e_d$ that happens
 before (\before) all conflicting events in $T_1$ and $T_2$ (other than $e_d$
 itself, if $e_d$ is a conflicting event). 
\label{lemma:decision}
\end{lemma}
\ifreport
\begin{proof}
\ICS{This Paragraph needs to be . . . better. Just better.}
Assume the contradiction.
Then for any full execution $E^\prime$ ending in $S$, an equivalent execution 
 exists featuring a state in which a conflicting event $e_c$ is scheduled, but 
 the decision event of $E^\prime$ is not.
Such an equivalent execution would by definition have a different decision 
 event, since $e_c$'s presence in a state makes the state univalent.
By our assumption, this equivalent execution has conflicting events that 
 neither are, nor happen after, its decision event.
This implies yet another equivalent execution with yet another state featuring 
 an even earlier conflicting event but not the decision event, and so on.
Since all states are finite sets, and \before is a strict partial order, this 
 infinite descending chain is impossible.
There must exist an execution $E$ ending with $S$ with decision event $e_d$ 
 that happens before all conflicting events in $T_1$ and $T_2$.
\end{proof}
\else
\begin{proof}[sketch]
We show that the contradiction implies an infinite chain of equivalent 
 executions with earlier and earlier non-decision conflicting events, which is 
 impossible given that system states are finite.
\end{proof}
\fi

We show that two fundamental system state properties are necessary for secure scheduling:
\begin{definition}[First-Precedes-Decision]
\label{def:first-precedes}
State $S$ satisfies \ti{First-Precedes-Decision} if, for any pair of
transactions $T_1$ and $T_2$ in $S$ with
 $T_1\before T_2$, there is a full execution $E$ ending in $S$ with a decision 
 event $e_d$ that either is in $T_1$, or happens after an event in $T_1$.
\end{definition}
\begin{definition}[Decision-Precedes-Second]
\label{def:precedes-second}
A state $S$ satisfies \ti{Decision-Precedes-Second} if, for any pair
of transactions $T_1$ and $T_2$ in $S$ with
 $T_1\before T_2$, there is a full execution $E'$ ending in $S$ with a
 decision event $e_d'$, such that no event in $T_2$ happens before 
 $e_d'$.
\end{definition}

Therefore, for a protocol to be \ProtocolSecurityAdjective, it must ensure
 resulting system states have these properties.

\begin{theorem}[Necessary Condition]
Any \ProtocolSecurityAdjective, dead\-lock-free protocol
 $p$ must ensure
 that all full executions consistent with $p$ feature only states
 satisfying both
 First-Precedes-Decision and Decision-Precedes-Second.
\label{theorem:necessary}
\end{theorem}
\begin{proof}
Given $T_1\before T_2$, any execution $E^\prime$ ending in $S$ features a 
 decision event $e_d$.
Decision events for the same pair of transactions in equivalent executions
 must agree on ordering, by the definition of equivalent execution.
If $T_1$ does not contain $E$'s decision event, $e_d$, or any event
 that happens before $e_d$, then there exists an
 equivalent execution in which $e_d$ is scheduled before any events in $T_1$ or
 $T_2$.
This execution would imply the existence of a system state in which no event 
 in either transaction is scheduled, but it is impossible to schedule $T_2$ 
 before $T_1$, regardless of inputs after that state. 
If, after this state, the start event for $T_2$ were scheduled, but not the 
 start event for $T_1$, then $T_2$ cannot be scheduled.
This contradicts a the deadlock-freedom requirement:
 no protocol should result in a system state in which 
 a supported transaction can never be scheduled.

Therefore some event in $T_1$ either is or happens before $e_d$ for some full 
 execution $E$ ending in $S$.

If $T_1$ and $T_2$ conflict, then $e_d^\prime$ either is an event in
 $T_1$ or happens before an event in $T_1$, by \cref{lemma:decision}.
If an event $e_2\in T_2$ happens before $e_d^\prime$, then either 
 $e_d^\prime\in T_1$, and
\[
e_2\before e_d^\prime\Rightarrow T_2\before T_1
\]
which is impossible, by the definition of happens-before, or \\${\exists e_1\in 
 T_1.  e_d^\prime\before e_1}$, and 
\[
e_2\before e_d^\prime\before e_1 \Rightarrow e_2\before e_1\Rightarrow T_2\before T_1
\]
which is also impossible, by the definition of happens-before.

If $T_1$ and $T_2$ do not conflict, then the only way $T_1\before T_2$ implies 
 that there exists some chain \ifreport \\ \fi $T_1\before T_3\before T_4\before \dots\before 
 T_n\before T_2$ such that and each transaction in the chain conflicts with 
 the next.
Therefore, by the above proof, an equivalent execution exists in which each 
 transaction in the chain contains the decision event for ordering itself and 
 the following transaction, and no events in the following transaction are 
 before that decision event.

Therefore there exists some equivalent execution $E^\prime$ in which no event 
 in $T_2$ happens before the decision event $e_d^\prime$ deciding the ordering 
 between $T_1$ and $T_2$. 
\end{proof} 

Although \cref{theorem:necessary} may seem trivial, it represents some 
 important conclusions: 
No protocol can make any final ordering decision until at least one 
 transaction involved has begun.
Furthermore, it is impossible for the later transaction to determine the decision.
Truly atomic transactions cannot include any kind of two-way interaction or 
 negotiation for scheduling.

\section{The Staged Commit Protocol}
\label{sec:protocols}
\label{sec:stagedcommit}

We now present the staged commit protocol (SC) and prove that it is
 \ProtocolSecurityAdjective, given transactions satisfying
 relaxed monotonicity.

SC is a hybrid of traditional serialization protocols,
 such as 2PC, and the simple pessimistic protocol described in
 the proof of \cref{theorem:monotonicity}. 
Compared to our simple pessimistic protocol, it allows a broader variety of
 transactions to be scheduled (relaxed monotonicity vs. regular monotonicity),
 which in turn allows more concurrency. 
A transaction is divided into \ti{stages}, each of which can be securely committed 
 using a more traditional protocol.
The stages themselves are executed in a pessimistic sequence.

Each event scheduled is considered to be either \ti{precommitted} or 
 \ti{committed}.
We express this in our model by the presence or absence of an ``isCommitted'' 
 event corresponding to every event in a transaction.
Intuitively, a \ti{precommitted} event is part of some ongoing transaction, so 
 no conflicting events that happen after a precommitted event should be 
 scheduled.
A \ti{committed} event, on the other hand, is part of a completed
 transaction;
 conflicting events that happen after a committed event can safely be scheduled.
Once an event is precommitted, it can never be un-scheduled. 
It can only change to being committed.
Once an event is committed, it can never change back to being precommitted.

\begin{itemize}
\item The events of each transaction are divided into stages.
      Each stage will be scheduled using traditional 2PC, so aborts
      within a
       stage will be sent to all locations involved in that stage.
      
      To divide the events into stages, we establish equivalence classes of
       the events' labels.
      Labels within each class are \ti{equivalent} in the following
      sense: when events with equivalent labels are aborted, those
      aborts can securely flow to the same set of locations.
      \JED{This raises the question of how one determines whether an
      abort can securely flow to a location. In the implementation,
      this is enforced by the pc constraint.}
      \JEDreply{Also, we could use an example here.}
      An event's abort can always flow to the event's own location, so
       locations involved in a stage can securely ensure the atomicity
       of the events in that stage.
      Since conflicting events have the same security labels, they will be in
       the same equivalence class.
      We call these equivalence classes \ti{conflict labels} (\texttt{cl}).

\item Each stage features events of the same conflict label, and is scheduled
       with 2PC.
      One location must coordinate the 2PC.
      All potential aborts in the stage must flow to the coordinator,
      and some events on
       the coordinator must be permitted to affect all events in the stage.
      Relaxed monotonicity implies that at least one such location
       exists for each conflict label.

      When a stage tries to schedule an event, but finds a precommitted
       conflicting event, it aborts the entire stage.
      Because conflicting events have the same label, these aborts
       cannot affect events on unpermitted locations.
      
      When a stage's 2PC completes, the events in the stage are scheduled, and
       considered \ti{precommitted}.

\item Each transaction precommits its stages as they occur.
      To avoid deadlock, we must ensure that whenever two transactions
       feature stages with equal conflict labels, they precommit those stages in
       the same order.
      Therefore, the staged commit protocol assumes an ordering of
       conflict labels. This can be any arbitrary ordering, so long as
       (1) it totally orders the conflict labels appearing in each
       transaction, and (2) all transactions agree on the ordering.


\item When all stages are precommitted, all events in the transaction can be
       committed.
      \ti{Commit} messages to this effect are sent between locations, backwards
       through the stages.
      Whenever an event in one stage triggers an event in the next, the
       locations involved can be sure a \ti{commit} message will take the
       reverse path.
      The only information conveyed is timing.
\end{itemize}

Because events in a precommitted stage cannot be
 un-scheduled or ``rolled back'',
a participant that is involved only in an earlier stage is prevented from gleaning
 any information about later stages.
The participant will only learn, eventually, that it can commit.


{\color{blue}\Patsy}'s transaction in \cref{fig:secure-hospital-positive} has at
 least two stages when the patient has HIV:
\begin{enumerate}
\item {\color{blue}\Patsy} begins the transaction ({\color{blue}\Patsy\ start}),
        and reads the address ({\color{blue}Read Address}).
      This stage will be atomically precommitted, and this precommit process
       will determine the relative ordering of {\color{blue}\Patsy}'s
       transaction and {\color{red}\Attacker}'s, independent of more secret events.
\item {\color{blue}\Patsy} finds that the patient has HIV ({\color{blue}Read
       HIV}), and prints the patient's address ({\color{blue}Print address}).
\end{enumerate}

\begin{theorem}[Security of SC]
\label{theorem:stagedcommit}
Any set of transactions satisfying relaxed
 monotonicity are serialized by SC securely without deadlock.
\end{theorem}
\ifreport
\begin{proof}\ \\
\else
\begin{proof}[sketch]
\fi
\noinunbo{Security}
SC preserves \relaxedsecurity.
Intuitively, any information flows that it adds are already included in the
 transaction.

SC adds no communication affecting security:
\begin{itemize}
\item Communication within each stage is strictly about events that all 
 participants can both observe.
\item For each pair of consecutive stages, at least one participant
from the first stage can notify a participant in the 
 second stage securely, when it is time for the second stage to begin.
Relaxed monotonicity ensures the second stage contains an event that happens
 after an event in the first stage, representing a line of communication.
\item Communication for commits can safely proceed in reverse order of stages.
\ifreport
Within each stage, each participant can securely forward a commit message to 
 all other participants.
Between stages, commit messages can be sent back along the same channels used 
 to notify each stage the previous one had precommitted.
\fi
Each participant knows when it precommits exactly which commit messages it 
 will receive.
\ifreport
The commit messages themselves do not leak any information (other than timing) 
 to their recipients.
\fi
\end{itemize}
\ifreport
Therefore SC adds no unauthorized information flows.

Specifically, for any given participant's label $\ell$, events within a stage 
 visible to $\ell$  are scheduled deterministically based only on information 
 visible to $\ell$.
Commit messages (and affiliated events) for visible stages arrive eventually, 
 at a time determined by network delay events, which we consider input.
Other stages' events are not observable to $\ell$.

Therefore, for any two executions beginning with states indistinguishable to 
 $\ell$, with NIEs visible to $\ell$, all scheduled events visible to $\ell$ 
 would be indistinguishable.
Thus \relaxedsecurity is preserved.
\fi

\noinunbo{Serializability}
\ifreport
Any set of transactions with relaxed monotonicity scheduled by SC
 will be serializable.
\begin{lemma}[Precommitted Snapshot]\ \\
\label{lemma:precommittedsnapshot}
Any
\else
Our proof is built around the following lemma: any
\fi
execution in which an event in a transaction is committed features a
 system state in which all events in the transaction are precommitted.
\ifreport
\end{lemma}
\begin{proof}
Stages are totally ordered, and each waits until the final stage commits 
 before $\p\before$ any of its events commit.
The final stage precommits before $\p\before$ it commits, and so there is 
 a system state in which all events in the transaction are precommitted.
\end{proof}

Let $E$ be an execution where any two conflicting transactions $T_1$ and $T_2$ 
 both have at least one event that commits.
Given \cref{lemma:precommittedsnapshot}, $E$ must feature two states:
 one in which all events in $T_1$ are precommitted, and another in which all 
 events of $T_2$ are precommitted.
As $T_1$ and $T_2$ conflict, these states cannot be the identical. 
(An event is never scheduled while a conflicting event is precommitted.)

One transaction must be scheduled before $\p\before$ the other.
Without loss of generality, let it be $T_1$.
No equivalent execution can feature a state in which an event in $T_2$ is 
 scheduled before an event in $T_1$, as this would require a conflicting event 
 in $T_2$ to be precommitted before its corresponding conflicting event in 
 $T_1$ is committed. 
The corresponding conflicting event in $T_1$ must be precommitted before any 
 event in $T_1$ commits, and we require that all events in $q_2$ remain 
 precommitted until after an event in $T_1$ commits. 

Therefore, if $T_1\before T_2$ then it is impossible for $T_2\before T_1$.
Thus 
\else
This lemma is used to show that
\fi
SC guarantees a strict partial order of transactions, and
 therefore serializability.

\noinunbo{Deadlock Freedom}
\ifreport

A deadlock can occur only if there is a cycle of dependencies among
transactions, in which transaction $T_1$ 
 \ti{depends on} $T_2$ if and only if $T_2$ has precommitted an event 
 conflicting with an unscheduled event in $T_1$.

Conflicting events share labels, and stages are defined by labels. 
All transactions must therefore order the stages of conflicting pairs in the same way.
One event can only ever depend on an event in its own or in a prior stage.
Stages are precommitted in order, so no dependency cycle featuring events in 
 different stages is possible.

Each stage is precommitted atomically using 2PC.
2PC preserves deadlock freedom, meaning no cycle featuring only events in the 
 same stage is possible.

Therefore no cycles, and thus no deadlock, can exist with SC.

SC is secure, deadlock-free, and guarantees serializability when
 the transactions have relaxed monotonicity.
\end{proof}
\else
Deadlock cannot form within any stage, since stages use 2PC, 
 which preserves deadlock freedom.
The stages themselves, like locks in our
proof of \cref{theorem:monotonicity},
are precommitted
 in a consistent order, guaranteeing deadlock freedom.
\end{proof}

\fi

\subsubsection*{The Importance of Optimism}

SC specifies only a commit protocol.
Actual computation (which generates the set of events) for each
transaction can be done
 in advance, optimistically.
 \ACM{Feel like we could use some citations to optimistic concurrency
 control here and, if possible, citations to some more recent DB
 literature making the point that optimism is a good idea for
 performance of distributed transactions.}
\ICSreply{agreed}
If one stage precommits and the next is blocked by a
 conflicting transaction, optimistically precomputed events would have to be 
 \ti{rolled back}.
However, no precommitted event need be rolled back.
In fact, it would be insecure to do so.
Thus SC allows for partially optimistic transactions with partial
 rollback.

Our model requires only that a transaction be a set of events.
In many cases, however, it is not possible to know which transaction will run
 when a start event is scheduled.
For example, a transaction might read a customer's banking information from a 
 database and contact the appropriate bank.
It would not be possible to know which bank should have an event in the 
 transaction beforehand.
If a system attempted to read the banking information prior to the 
 transaction, then serializability is lost: the customer might change banks in 
 between the read and the transaction, and so one might contact the wrong bank.

Optimism solves this problem: 
events are precomputed, and when an entire stage is completed, that stage's 
 2PC begins.
This means that optimism is not just an optimization; it is required for
 secure scheduling in cases where the transactions' events are not known
 in advance.

\section{Implementation}
\label{sec:implementation}
We extended the Fabric language and compiler to check that
transactions can be securely scheduled, and we extended the Fabric
runtime system to use SC.
Fabric and IFDB~\cite{schultz2013} are the two open-source systems we are aware
 of that support distributed transactions on persistent, labeled data with
 information flow control.
Of these, we chose Fabric for its static reasoning capabilities.
IFDB checks labels entirely dynamically, so it cannot tell if a 
 transaction is schedulable until after it has begun.

\subsection{The Fabric Language}
The Fabric language is designed for writing distributed programs using atomic 
 transactions that operate on persistent, Java-like objects~\cite{fabric09}.
It has types that label each object field with information flow
 policies for confidentiality and integrity.
The compiler uses these labels to check that Fabric programs enforce a
 noninterference property.
However, like all modern systems built using 2PC, Fabric does not require that 
 transactions be securely scheduled according to the policies in the program.
Consequently, until now, abort channels have existed in Fabric.
\ACM{Can we soften this a little? Presumably only multi-store
transactions have been subject to abort channels. So it has been
possible to write transactions that create abort channels, certainly.
Do we know of any actual programs that have abort channels?}
\ICSreply{The hospital example is a single-store transaction featuring an abort
 channel.
As for pre-existing examples, I want to say the travel agency one had something
 where one airline could leak some bits to another?}
 \TMreply{Not sure if it would be repetitive or helpful but we could
 point out the hospital example here as evidence.  Would that help
 Andrew?}

We leverage these security labels and extend the compiler to additionally 
 check that transactions in a Fabric program are monotonic 
 (\cref{sec:analysis}).
This implementation prevents confidentiality breaches
 via abort channels. Preventing integrity breaches would require
 further dynamic checks, which we leave to future work.
\TM{I wasn't a fan of the original wording and wrote this sentence as a
 replacement, but I'm not sure this is all that much better.  Here's the
 original in case we want to revert: We believe our implementation to be
 secure against confidentiality breaches, but have not ruled out all
 integrity breaches.}
\TMreply{Ugh, I feel like as a reader I want to be at least given a hint
 about what's going on with the integrity checks.}

%

\subsection{Checking Monotonicity}

Our modification to the Fabric compiler enforces relaxed monotonicity
 (\cref{definition:relaxedmonotonicity}).
\JED{Changed this. Used to say "strict monotonicity", which I don't
believe to be true.}
Our evaluation (\cref{sec:evaluation}) shows that enforcing this
 condition does not exclude realistic and desirable programs.
Our changes to the Fabric compiler and related files include
 4.1k lines of code (out of roughly 59k lines).

\subsubsection{Events and Conflict Labels in Fabric}
\label{sec:conflictlabels}
The events in the system model~(\cref{sec:system}) are represented in
our implementation by read and writes on
 fields of persistent Fabric objects.
The label of the field being read or written
 corresponds to the event labels in our model.

SC~(\cref{sec:stagedcommit}) divides events
into stages based on conflict labels (\texttt{cl}).
In our implementation, we define the \texttt{cl}
\TM{Thought: do we want to call these conflict classes (or some other
 term without the word label) to make the distinction between event
 label and these classes clearer?}
of an event $e$ to correspond to the set of \ti{principals} authorized to
read or write the field that is being accessed by $e$.
If $e$ is a write event, this set contains exactly those principals
that can perform a conflicting operation (and thereby receive an
abort);
if $e$ is a read event,
the set is a conservative over-approximation, since only the writers can
conflict.

\cref{fig:blog} presents a program in which Carol schedules two
events within a single transaction. First, she
reads a blog post with security label $\ell$. Second, she writes a
comment (whose content depends on that of the post) with label
$\ell^\prime$.
Since $\ell$ permits Alice, Bob, or Carol to read the post,
the \texttt{cl} of the first event
 includes all three principals.
However, only Alice and Carol can read or write the comment, so when Carol goes
 to write it, only Alice or another transaction acting on behalf of Carol could
 cause conflicts.
The \texttt{cl} of the write therefore includes only Alice and Carol.

\begin{figure}[t]
\centering
\lstset{
    literate={B}{{\ensuremath\bot}}1
             {L}{{\ensuremath\ell}}1
             {M}{{\ensuremath{\ell^\prime}}}1
             {H}{{\tb{PC}}}1
             {J}{{\tb{Possible conflictors}}}1
             {A}{{$\cb{Alice,Bob,Carol}$}}1
             {K}{{$\cb{Alice,Carol}$}}1
             {N}{{-}}1
    }
\begin{lstlisting}[gobble=2]
  atomic {                        H   J
    String{L} p = post.read();    B  A
    Comments{M} c;                B     N
    if (p.contains("fizz")) {     B     N
      c.write("buzz");            L   K
    if (p.contains("buzz")) {     B     N
      c.write("fizz");            L   K
    }
  }
\end{lstlisting}
\vspace{-3mm}
\caption{Carol's program in our Blog example:
         Carol reads a post with label $\ell$, and depending on what she 
          reads, writes a comment with label $\ell^\prime$.
         Label $\ell$ permits Alice, Bob, and Carol to read the post,
          while $\ell^\prime$ keeps the Comments more 
          private and allows only Alice and Carol to view or edit.}
\label{fig:blog}
\end{figure}

\subsubsection{Program Counter Label}

The \ti{program counter} label (\texttt{pc})~\cite{denning-book}
labels the program context. For
 any given point in the code,
the \texttt{pc} represents the join (least upper bound) of the labels of events
that determine
 whether or not execution reaches that point in the code. These events
 include those occurring in if-statement and loop conditionals.
For instance, in \cref{fig:blog}, whether line 5 runs depends on the value of
 \texttt p, which has label $\ell$.
Therefore, the fact that line 5 is executing is as secret as \texttt p, and the
 \texttt{pc} at line 5 is $\ell$.

SC requires that when events with the same
\texttt{cl} are aborted, those aborts can securely flow to the same
set of locations.
\JED{This requirement is buried in the protocol description, and the
reader surely has forgotten about it by now. Unsure how to help this.
Maybe give it a name?}
When an event causes an abort, the resulting abort messages carry
information about the context in which the event occurs.
Therefore, we enforce the requirement by introducing a
constraint on the program context in which events may occur: the
\texttt{pc} must flow to the principals in the conflict label.
\begin{equation}
  \label{eq:pc-constraint}
  \texttt{pc}\ \less\ \texttt{cl}
\end{equation}

Eliding the details of how Fabric's labels are structured, in \cref{fig:blog},
 $\bot$ flows to everything, and $\ell$, the label
 of the blog post, does flow to the conflict label, indicating that both Alice
 and Carol can cause a conflict.
Therefore, \cref{eq:pc-constraint} holds on lines 2, 5, and 7.

\subsubsection{Ordering Stages}
\label{sec:stages}
Each stage consists of operations with the same \texttt{cl}.
To ensure all transactions precommit conflicting stages in the same 
 order, we adopt a universal stage ordering:
\begin{equation}
  principals\p{\texttt{cl}_{i}}\supsetneq principals\p{\texttt{cl}_{i+1}}
\end{equation}
The set of principals in each stage must be a strict superset of the
 principals in the next one.
This ensures that unrestricted information can be read in one
 stage and sensitive information can be modified in a later stage in
 the same transaction.
In the hospital example (\cref{fig:secure-hospital}), {\color{blue}Read HIV} 
 has a conflict label that only includes trusted personnel, while
 {\color{blue}Read address} has a conflict label that includes more hospital staff.
As a result, our implementation requires that {\color{blue}Read
address} be staged before {\color{blue}Read
 HIV} in \Patsy's transaction.

In \cref{fig:blog}, our stage ordering means that the read on line 2, with a
 \texttt{cl} of $\cb{Alice, Bob, Carol}$ belongs in an earlier stage than the
 write, which features a \texttt{cl} of only $\{Alice, Carol\}$.

\subsubsection{Method Annotations}
\label{sec:method-annotations}
To ensure modular program analysis and compilation, each method
 is analyzed independently.
Fabric is an object-oriented language with dynamic dispatch, so it is not
 always possible to know in advance which method implementation a program will 
 execute.
Therefore, the exact conflict labels for events within a method call are not 
 known at compile time.
In order to ensure each atomic program can divide into monotonic stages, we 
 annotate each method with bounds on the conflict labels of operations within 
 the method.
These annotations are the security analogue of argument and return
types for methods.

\subsection{Implementing SC}
\label{sec:implementation-runtime}
We extended the Fabric runtime system to use SC
 instead of traditional 2PC,
modifying 2.4k lines of code out of a total of 24k lines
 of code in the original implementation.
Specifically, we changed Fabric's 2PC-based transaction protocol so that it 
 leaves each stage prepared until all stages are ready, and then commits.

Since Fabric labels can be dynamic, the compiler statically determines 
 \ti{potential stagepoints}---points in the program that may begin a new
 stage---along with the conflict labels of the stages immediately
 surrounding the potential stagepoint.
If the compiler cannot statically determine whether the conflict labels
 before and after a stagepoint will be different, it inserts a dynamic
 equivalence check for the two labels.
At run time, if the two labels are not equivalent, then a stage is
ending, and the system precommits all operations made thus far.
To precommit a stage, we run the first (``prepare'') phase of
 2PC.
If there is an abort, the stage is re-executed until it eventually
 precommits.

In \cref{fig:blog}, there is a potential stagepoint
 before lines 4 and 6, where the next operation in each case will not
 include Bob as a possible conflictor.
The conflict labels surrounding the potential stagepoint are
 $\cb{Alice,Bob,Carol}$ (from reading the post on line 2) and
 $\cb{Alice, Carol}$ (from writing the comment on either line 4 or 6).
If another transaction caused the first stage to abort, then Carol's code 
 would rerun up to line 4 or 6 until it could precommit, and then the
 remainder of the transaction would run.

\section{Evaluation}
\label{sec:evaluation}
To evaluate our implementation, we built three example Fabric applications, and 
 tested them using our modified Fabric compiler:
\begin{itemize}\itemsep 0in
  \item an implementation of the hospital example from 
         \cref{sec:abortchannels};
  \item a primitive blog application (from which \cref{fig:blog} was
         taken), in which participants write and
         comment on posts with privacy policies; and
  \item an implementation of the Rainforest example from
         \cref{sec:abortchannels}.
\end{itemize}

\subsection{Hospital}
We implemented the programs described in
 our hospital example (\cref{fig:insecure-hospital}).
In the implementation, \Patsy's code additionally appends the
 addresses of HIV-positive patients to a secure log. In a third
 program, another trusted participant reads the secure log.

With our changes, the compiler correctly rejects \Patsy's code.
We amended her code to reflect \cref{fig:secure-hospital}.
Of the 350 lines of code, we had to change a 
 total of 113 to satisfy relaxed monotonicity and compile.
Of these 113 lines, 23 were additional method annotations and the
 remaining 90 were the result of refactoring the transaction that
 retrieves the addresses of HIV-positive patients.
SC scheduled the transactions without leaking information.
The patient's HIV status made
 \Attacker\ neither more nor less likely to receive aborts.

\subsection{Blog}
In our primitive blog application, a store holds API objects,
 each of which features blog posts (represented as strings) with some security
 label, and comments with another security label.
These labels control who can view, edit, or add to the posts and comments.

In one of our programs, the blog owner atomically reads a post and
updates its text
to alternate between ``fizz'' and
 ``buzz''.
In another program, another user comments on the first post 
 (\cref{fig:blog}).
To keep this comment pertinent to the content of the post, reading the post 
 and adding the comment are done atomically.
Since posts and comments have different labels, this transaction has at least two
 stages:
one to read the post, and another to write the comment.

We were able to compile and run these programs with our modified system with 
 relatively few changes.
Of the 352 lines of code, we had to change a total of 50,
primarily by adding annotations
 to method signatures~(\cref{sec:method-annotations}).

\subsection{Rainforest}
\label{sec:evaluation-rainforest}
\begin{figure}
\centering
\begin{tabular}{|l|c|c|}\hline
  \multicolumn{1}{|c|}{\tb{Data item}} & \tb{Readers} & \tb{Writers} \\\hline
  Gloria's account balance             & Bank, Gloria & Bank \\\hline
  Item price                           & (public)     & Outel \\\hline
  Inventory                            & Outel        & Outel \\\hline
\end{tabular}
\caption{Example policies for the Rainforest application.}
\label{fig:rainforest-policies}
\end{figure}

We implemented the Rainforest example from \cref{sec:rainforest}.
In our code, two nodes within Rainforest act with Rainforest's
authority.
They perform transactions representing the orders of Gloria and Fred from
 \cref{fig:rainforest}.
Each transaction updates inventory data stored at one location, and banking 
 data stored at another.
\cref{fig:rainforest-policies} gives examples of the policies
for price, inventory, and banking data.

\newcommand{\resultstable}[0] {{
\begin{figure*}
\centering
\begin{tabular}{| r | l | c | c | c | c |}\hline
  \multicolumn{1}{|c|}{\multirow{2}{*}{\tb{Example}}}&
  \multicolumn{1}{|c|}{\multirow{2}{*}{\tb{Program}}}&
  \multicolumn{3}{|c|}{\tb{SC}}&
  \tb{2PC}\\ \cline{3-6}
  &&
  \tb{\# stages}&
  \tb{Dyn. checks}&
  \tb{Total time}&
  \tb{Total time}\\
  \hline
%
%

Hospital & \texttt{patsy}             & 3 & 0.45 ms & 9.17 ms & 6.38 ms \\\hline
\multirow{2}{*}{Blog} & \texttt{post} & 2 & 0.11 ms & 1.03 ms & 1.01 ms \\\cline{2-6}
& \texttt{comment}                    & 3 & 0.29 ms & 1.30 ms & 1.01 ms \\\hline
\end{tabular}
\caption{Performance overhead of SC. Reported times are
per-transaction averages, across
three 5-minute runs of the blog application and three
20-minute
runs of the hospital application. Relative standard error of all
measurements is less than 2\%.
}
\label{fig:overhead}
\end{figure*}
}}
\ifreport
\else
  \resultstable{}
\fi

While attempting to modify this code to work with SC, we discovered
that the staging order chosen in \cref{sec:stages} makes it impossible
to provide the atomicity of the original application while both
meeting its security requirements and ensuring deadlock freedom.

To illustrate, suppose Gloria is purchasing an item from Outel.
To ensure she is charged the correct price, the event that updates the inventory
 must share a transaction with the one that debits Gloria's bank account.
The conflict label for the inventory event corresponds to $\{\tt{Outel}\}$,
 whereas the conflict label for the debit event corresponds to $\{\tt{Bank},
 \tt{Gloria}\}$.
Since neither is a subset of the other, the compiler cannot
 put them in the same transaction.

These difficulties in porting the Rainforest application arise because
 Fabric is designed to be an open system, and so an \ti{a priori}
 choice of staging order must be chosen.
If the application were written as part of a closed system, deadlock freedom
 can be achieved by picking a staging order that works for this particular
 application (e.g., $\{\tt{Outel}\}$ before $\{\tt{Bank,Gloria}\}$),
 but it might be difficult to extend the system with future
 applications.

\subsection{Overhead}
\label{sec:overhead}

The staged commit protocol adds two main sources of overhead compared to
 traditional 2PC.
First, each stage involves a round trip to prepare the data manipulated
 during the stage, leading to overhead that scales with the number of
 stages and with network latency.
Second, as described in \cref{sec:implementation-runtime}, dynamic
labels result in potential stagepoints, which must be resolved using
run-time checks.
The number of checks performed depends on how
 well the compiler's static analysis predicts
 potential stagepoints.

We measured this overhead in our implementation on an
 Intel Core i7-2600 machine with 16 GiB of memory, using the transactions
 in our examples.
The post and comment transactions in the blog example were each run
continually for 15 minutes, and
 \Patsy's transaction in the hospital example was run continually
 for 1 hour.

\cref{fig:overhead} gives the overall execution times for both the
original system and the modified system. For the modified system, it
also shows the number of
stages for each transaction and the average time spent in
dynamic checks for resolving potential stagepoints.
The \texttt{comment} transaction in our
 experiments has one more stage than as described in \cref{fig:blog},
 because in all transactions, there is an initial stage performed
 to obtain the principals involved in the application.

By running the nodes on a single machine and using in-memory data storage, we
 maximize the fraction of the transaction run time occupied by dynamic checks.
Nevertheless, this fraction remains small.
\TM{Do we want to point to the hospital numbers as good evidence of this
since it's the largest transaction of the three?}
While the effective low latency of communication between nodes reduces the
 overhead due to communication round-trips for staging precommits, we report the
 number of stages, from which this overhead can be calculated for arbitrary
 latency.
\ifreport
  \resultstable{}
\fi

\section{Related work}
\label{sec:related}
Various goals for atomic transactions, such as
serializability~\cite{Papa79} and ACID~\cite{Haerder1983},
have long been proposed and widely studied, and
are still an
active research topic%
\ifreport~\cite{raz92,Kang1995,Smith1996,fabric09,Avni2015,Cerone2015}\fi.
While much of the recent interest has been focused on
 performance\ifreport~\cite{dragojevic2015,Lee2015,Wei2015,Aguilera2015,Zhang2015,Xie2015}\fi,
 we focus on security.

Information leaks in commonly used transaction scheduling protocols have been known for at least two
decades~\cite{Smith1996,atluri1996}.
Kang and Keefe~\cite{Kang1995} explore transaction processing in databases with multiple
 security levels.
Their work focuses on a simpler setting with a global, trusted transaction
 manager. They assume each transaction has a single security level, and can only
 ``read down'' and ``write up.''
%
Smith et al.~\cite{Smith1996} show that strong atomicity, isolation,
and consistency guarantees are not possible for all transactions in a
generalized multilevel secure database.
They propose weaker guarantees and give three different
 protocols that meet various weaker guarantees.
Their Low-Ready-Wait 2PL protocol is similar to SC,
 and provides only what the authors call ACIS$^-$--correctness.
Specifically, ``aborted operations at a higher level may prevent all lower level
 operations from beginning''~\cite[p37]{Smith1996}.
Although our implementation is conservative and would not allow such a thing,
 the theory behind SC could allow a later stage with less trustworthy
 participants to hold up earlier, precommitted stages indefinitely.
Duggan and Wu~\cite{duggan2011} observe that aborts in high-security subtransactions can leak
 information to low-security parent transactions.
Their model of a single, centralized multilevel secure database with
 strictly ordered security levels is more restrictive than our distributed model
 and security lattice. 
Our abort channels generalize their observation.
They arrive at a different solution, building a theory of secure nested
 transactions.
Atluri, Jajodia, and George~\cite{atluri2000} describe a number of known protocols requiring weaker
 guarantees or a single trusted coordinator.
Our work instead focuses on securely serializing transactions
in a fully decentralized setting.
Our analysis is also the first in this vein to consider liveness: SC can guarantee deadlock freedom of transactions with relaxed monotonicity.

In this work, we build on a body of research that uses lattice-based
information flow labels and language-based information flow
methods~\cite{denning-lattice,denning-cert,sm-jsac}.
Relatively little work has studied information flow in transactional
systems.  Our implementation is built on
Fabric~\cite{fabric09,fabric-release03}, a distributed programming
system that controls information flow over persistent objects.
The only other information-flow-sensitive database implementation
appears to be IFDB~\cite{schultz2013}, which also does not account for abort
channels.
\ACM{Hmm, what about those MLS databases? Seem to remember Jajodia,
Denning working on that.}
\ICSreply{most of those papers seem to have been theoretical.
I assume someone must have implemented an MLS at some point...}

\section{Conclusion}
\label{sec:conclusion}
There is a fundamental trade-off between strong consistency guarantees and 
 strong security properties in decentralized systems.
We investigate the secure scheduling of transactions, a ubiquitous building 
 block of modern large-scale applications.
Abort channels offer a stark example of an unexplored security flaw:
existing transaction scheduling mechanisms can leak confidential information, 
 or allow unauthorized influences of trusted data.
While some sets of transactions are impossible 
 to serialize securely, we demonstrate the viability of secure
 scheduling.
 \ACM{Too strong? Not clear we've really demonstrated it for "many" applications.}
 \ICSreply{Yeah, but "some applications" sounds super weak.}
 \JEDreply{better?}

We present relaxed monotonicity, a simple condition
under which secure scheduling is always possible.
Our staged commit protocol can securely schedule any set of transactions with relaxed 
 monotonicity, even in an open system.
To demonstrate the practical applicability of this protocol, we adapted 
 the Fabric compiler to check transactional programs for conditions that allow 
 secure scheduling.
These checks are effective: the compiler identifies an intrinsic security flaw
 in one program, and accepts other, secure transactions with minimal adaptations.

This work sheds light on the fundamentals of secure transactions.
However, there is more work to be done to understand the pragmatic implications.
We have identified separate necessary and sufficient conditions for secure 
 scheduling, but there remains space between them to explore.
\ACM{Is work needed on how to deal with problems like the one we ran
into in Rainforest?}
\ICSreply{Maybe?}
Ultimately, abort channels are just one instance of the general problem of 
 information leakage in distributed systems.
Similar channels may exist in other distributed settings,
\ACM{I think this can be more punchy. "Distributed settings" is very
mushy. Isn't the point that distributed protocols leak information and
there has not been much work on studying how to identify such leaks or
prevent them?}
and we expect it to be fruitful to explore other protocols through
 the lens of information flow analysis.


\ifacknowledgments
\section*{Acknowledgments}
The authors would like to thank the anonymous reviewers for their
suggestions.  This work was supported by MURI grant FA9550-12-1-0400,
by NSF grants 1513797, 1422544, 1601879, by gifts from Infosys and
Google, and by the Department of Defense (DoD) through the
National Defense Science \& Engineering Graduate Fellowship (NDSEG)
Program.
\fi

\ifreport
  \bibliographystyle{abbrvurl}
  \bibliography{../bibtex/pm-master}
\else
  \bibliographystyle{abbrv}
  \small
  \bibliography{shortened,../bibtex/pm-master}

\begin{thebibliography}{10}

\bibitem{NET-transactions}
Distributed transactions: {.NET} framework 4.6.
\newblock
  \url{https://msdn.microsoft.com/en-us/library/ms254973%28v=vs.110%29.aspx}.
\newblock Accessed: 2015-11-13.

\bibitem{xastandard}
{XA} standard.
\newblock In L.~Liu and M.~T. Özsu, editors, {\em Encyclopedia of Database
  Systems}, pages 3571--3571. Springer US, 2009.
\newblock URL: \url{http://dx.doi.org/10.1007/978-0-387-39940-9_4060}, \href
  {http://dx.doi.org/10.1007/978-0-387-39940-9_4060}
  {\path{doi:10.1007/978-0-387-39940-9_4060}}.

\bibitem{Aguilera2015}
M.~K. Aguilera, J.~B. Leners, and M.~Walfish.
\newblock Yesquel: Scalable sql storage for web applications.
\newblock In {\em Proceedings of the 25th Symposium on Operating Systems
  Principles}, SOSP '15, pages 245--262, New York, NY, USA, 2015. ACM.
\newblock URL: \url{http://doi.acm.org/10.1145/2815400.2815413}, \href
  {http://dx.doi.org/10.1145/2815400.2815413}
  {\path{doi:10.1145/2815400.2815413}}.

\bibitem{fabric-release03}
O.~Arden, J.~Liu, T.~Magrino, and A.~C. Myers.
\newblock {Fabric 0.3}.
\newblock Software release, \url{http://www.cs.cornell.edu/projects/fabric},
  June 2016.
\newblock URL: \url{https://www.cs.cornell.edu/projects/fabric}.

\bibitem{azm10}
A.~Askarov, D.~Zhang, and A.~C. Myers.
\newblock Predictive black-box mitigation of timing channels.
\newblock In {\em 17\textsuperscript{th} ACM Conf.~on Computer and
  Communications Security (CCS)}, pages 297--307, Oct. 2010.
\newblock URL: \url{http://www.cs.cornell.edu/andru/papers/timing.html}.

\bibitem{atluri2000}
V.~Atluri, S.~Jajodia, and B.~George.
\newblock {\em Multilevel Secure Transaction Processing}.
\newblock Advances in Database Systems. Springer US, 2000.
\newblock URL: \url{https://books.google.com/books?id=5bsZAQAAIAAJ}.

\bibitem{atluri1996}
V.~Atluri, S.~Jajodia, T.~F. Keefe, C.~D. McCollum, and R.~Mukkamala.
\newblock Multilevel secure transaction processing: Status and prospects.
\newblock {\em DBSec}, 8(1):79--98, 1996.
\newblock URL:
  \url{http://citeseerx.ist.psu.edu/viewdoc/download?doi=10.1.1.51.89&rep=rep1&type=pdf}.

\bibitem{Avni2015}
H.~Avni, E.~Levy, and A.~Mendelson.
\newblock Hardware transactions in nonvolatile memory.
\newblock In Y.~Moses, editor, {\em Distributed Computing}, volume 9363 of {\em
  Lecture Notes in Computer Science}, pages 617--630. Springer Berlin
  Heidelberg, 2015.
\newblock URL: \url{http://dx.doi.org/10.1007/978-3-662-48653-5_41}, \href
  {http://dx.doi.org/10.1007/978-3-662-48653-5_41}
  {\path{doi:10.1007/978-3-662-48653-5_41}}.

\bibitem{barthe2006}
G.~Barthe, T.~Rezk, and M.~Warnier.
\newblock Preventing timing leaks through transactional branching instructions.
\newblock {\em Electron. Notes Theor. Comput. Sci.}, 153(2):33--55, May 2006.
\newblock URL:
  \url{https://www-sop.inria.fr/lemme/Tamara.Rezk/publication/Barthe-Rezk-Warnier.pdf},
  \href {http://dx.doi.org/10.1016/j.entcs.2005.10.031}
  {\path{doi:10.1016/j.entcs.2005.10.031}}.

\bibitem{Bertino:2001}
E.~Bertino, B.~Catania, and E.~Ferrari.
\newblock A nested transaction model for multilevel secure database management
  systems.
\newblock {\em ACM Trans. Inf. Syst. Secur.}, 4(4):321--370, Nov. 2001.

\bibitem{calder2011}
B.~Calder, J.~Wang, A.~Ogus, N.~Nilakantan, A.~Skjolsvold, S.~McKelvie, Y.~Xu,
  S.~Srivastav, J.~Wu, H.~Simitci, et~al.
\newblock Windows {Azure} {Storage}: a highly available cloud storage service
  with strong consistency.
\newblock In {\em 23\textsuperscript{rd} {ACM} Symp.~on Operating System
  Principles (SOSP)}, pages 143--157. ACM, 2011.

\bibitem{Cerone2015}
A.~Cerone, A.~Gotsman, and H.~Yang.
\newblock Transaction chopping for parallel snapshot isolation.
\newblock In Y.~Moses, editor, {\em Distributed Computing}, volume 9363 of {\em
  Lecture Notes in Computer Science}, pages 388--404. Springer Berlin
  Heidelberg, 2015.
\newblock URL: \url{http://dx.doi.org/10.1007/978-3-662-48653-5_26}, \href
  {http://dx.doi.org/10.1007/978-3-662-48653-5_26}
  {\path{doi:10.1007/978-3-662-48653-5_26}}.

\bibitem{cs08}
M.~R. Clarkson and F.~B. Schneider.
\newblock Hyperproperties.
\newblock In {\em {IEEE} Symp.~on Computer Security Foundations (CSF)}, pages
  51--65, June 2008.

\bibitem{corbett2013}
J.~C. Corbett, J.~Dean, M.~Epstein, A.~Fikes, C.~Frost, J.~J. Furman,
  S.~Ghemawat, A.~Gubarev, C.~Heiser, P.~Hochschild, et~al.
\newblock Spanner: Google’s globally distributed database.
\newblock {\em ACM Transactions on Computer Systems (TOCS)}, 31(3):8, 2013.

\bibitem{denning-lattice}
D.~E. Denning.
\newblock A lattice model of secure information flow.
\newblock {\em Comm.~of the ACM}, 19(5):236--243, 1976.

\bibitem{denning-book}
D.~E. Denning.
\newblock {\em Cryptography and Data Security}.
\newblock Addison-Wesley, Reading, Massachusetts, 1982.
\newblock URL:
  \url{http://www.amazon.com/Cryptography-Security-Dorothy-Elizabeth-Robling/dp/0201101505}.

\bibitem{denning-cert}
D.~E. Denning and P.~J. Denning.
\newblock Certification of programs for secure information flow.
\newblock {\em Comm.~of the ACM}, 20(7):504--513, July 1977.
\newblock URL: \url{http://dl.acm.org/citation.cfm?id=359712}.

\bibitem{dragojevic2015}
A.~Dragojevi\'{c}, D.~Narayanan, E.~B. Nightingale, M.~Renzelmann, A.~Shamis,
  A.~Badam, and M.~Castro.
\newblock No compromises: Distributed transactions with consistency,
  availability, and performance.
\newblock In {\em 25\textsuperscript{th} {ACM} Symp.~on Operating System
  Principles (SOSP)}, pages 54--70, New York, NY, USA, 2015. ACM.
\newblock URL: \url{http://doi.acm.org/10.1145/2815400.2815425}, \href
  {http://dx.doi.org/10.1145/2815400.2815425}
  {\path{doi:10.1145/2815400.2815425}}.

\bibitem{duggan2011}
D.~Duggan and Y.~Wu.
\newblock Transactional correctness for secure nested transactions - (extended
  abstract).
\newblock In {\em Trustworthy Global Computing - 6th International Symposium,
  {TGC} 2011, Aachen, Germany, June 9-10, 2011. Revised Selected Papers}, pages
  179--196, 2011.
\newblock URL:
  \url{https://www.cs.purdue.edu/transact11/web/papers/Duggan.pdf}, \href
  {http://dx.doi.org/10.1007/978-3-642-30065-3_11}
  {\path{doi:10.1007/978-3-642-30065-3_11}}.

\bibitem{2phase}
K.~P. Eswaran, J.~N. Gray, R.~A. Lorie, and I.~L. Traiger.
\newblock The notions of consistency and predicate locks in a database system.
\newblock {\em Comm.~of the ACM}, 19(11):624--633, Nov. 1976.
\newblock Also published as IBM RJ1487, December 1974.

\bibitem{Fischer82b}
M.~J. Fischer, N.~A. Lynch, and M.~S. Paterson.
\newblock Impossibility of distributed consensus with one faulty process.
\newblock {\em Journal of the ACM}, 32(2):374--382, Apr. 1985.
\newblock Also published as MIT Laboratory of Science Technical Report
  MIT/LCS/TR-282, Cambridge, MA, 1982.

\bibitem{GM82}
J.~A. Goguen and J.~Meseguer.
\newblock Security policies and security models.
\newblock In {\em IEEE Symp.~on Security and Privacy}, pages 11--20, Apr. 1982.

\bibitem{Haerder1983}
T.~Haerder and A.~Reuter.
\newblock Principles of transaction-oriented database recovery.
\newblock {\em ACM Comput. Surv.}, 15(4):287--317, Dec. 1983.
\newblock URL: \url{http://doi.acm.org/10.1145/289.291}, \href
  {http://dx.doi.org/10.1145/289.291} {\path{doi:10.1145/289.291}}.

\bibitem{Kang1995}
I.~E. Kang and T.~F. Keefe.
\newblock Transaction management for multilevel secure replicated databases.
\newblock {\em J. Comput. Secur.}, 3(2-3):115--145, Mar. 1995.
\newblock URL: \url{http://dl.acm.org/citation.cfm?id=2699799.2699802}.

\bibitem{Kopf:Durmuth:CSF2009}
B.~K{\"o}pf and M.~D{\"u}rmuth.
\newblock A provably secure and efficient countermeasure against timing
  attacks.
\newblock In {\em 2009 IEEE Computer Security Foundations}, July 2009.

\bibitem{Lamport78}
L.~Lamport.
\newblock Time, clocks, and the ordering of events in a distributed system.
\newblock {\em Comm.~of the ACM}, 21(7):558--565, July 1978.

\bibitem{Lee2015}
C.~Lee, S.~J. Park, A.~Kejriwal, S.~Matsushita, and J.~Ousterhout.
\newblock Implementing linearizability at large scale and low latency.
\newblock In {\em Proceedings of the 25th Symposium on Operating Systems
  Principles}, SOSP '15, pages 71--86, New York, NY, USA, 2015. ACM.
\newblock URL: \url{http://doi.acm.org/10.1145/2815400.2815416}, \href
  {http://dx.doi.org/10.1145/2815400.2815416}
  {\path{doi:10.1145/2815400.2815416}}.

\bibitem{fabric09}
J.~Liu, M.~D. George, K.~Vikram, X.~Qi, L.~Waye, and A.~C. Myers.
\newblock Fabric: A platform for secure distributed computation and storage.
\newblock In {\em 22\textsuperscript{nd} {ACM} Symp.~on Operating System
  Principles (SOSP)}, pages 321--334, Oct. 2009.
\newblock URL: \url{http://www.cs.cornell.edu/andru/papers/fabric-sosp09.html}.

\bibitem{mccu88}
D.~McCullough.
\newblock Noninterference and the composability of security properties.
\newblock In {\em IEEE Symp.~on Security and Privacy}, pages 177--186. IEEE
  Press, May 1988.

\bibitem{java-beans}
S.~Microsystems.
\newblock {JavaBeans} (version 1.0.1-a).
\newblock \url{http://java.sun.com/products/javabeans/docs/spec.html}, Aug.
  1997.

\bibitem{ml-tosem}
A.~C. Myers and B.~Liskov.
\newblock Protecting privacy using the decentralized label model.
\newblock {\em ACM Transactions on Software Engineering and Methodology},
  9(4):410--442, Oct. 2000.
\newblock URL: \url{http://www.cs.cornell.edu/andru/papers/iflow-tosem.pdf}.

\bibitem{Papa79}
C.~H. Papadimitriou.
\newblock The serializability of concurrent database updates.
\newblock {\em Journal of the ACM}, 26(4):631--653, Oct. 1979.

\bibitem{ports2012}
D.~R.~K. Ports and K.~Grittner.
\newblock Serializable snapshot isolation in {PostgreSQL}.
\newblock {\em Proc. VLDB Endow.}, 5(12):1850--1861, Aug. 2012.
\newblock URL: \url{http://dx.doi.org/10.14778/2367502.2367523}, \href
  {http://dx.doi.org/10.14778/2367502.2367523}
  {\path{doi:10.14778/2367502.2367523}}.

\bibitem{raz92}
Y.~Raz.
\newblock The principle of commitment ordering, or guaranteeing serializability
  in a heterogeneous environment of multiple autonomous resource managers using
  atomic commitment.
\newblock In {\em 18th Very Large Data Bases Conference ({VLDB})}, Aug. 1992.

\bibitem{Roscoe95}
A.~W. Roscoe.
\newblock {CSP} and determinism in security modelling.
\newblock In {\em IEEE Symp.~on Security and Privacy}, pages 114--127, May
  1995.
\newblock URL: \url{http://dx.doi.org/10.1109/SECPRI.1995.398927}, \href
  {http://dx.doi.org/10.1109/SECPRI.1995.398927}
  {\path{doi:10.1109/SECPRI.1995.398927}}.

\bibitem{sm-jsac}
A.~Sabelfeld and A.~C. Myers.
\newblock Language-based information-flow security.
\newblock {\em IEEE Journal on Selected Areas in Communications}, 21(1):5--19,
  Jan. 2003.
\newblock URL: \url{http://www.cs.cornell.edu/andru/papers/jsac/sm-jsac03.pdf}.

\bibitem{sm04}
A.~Sabelfeld and A.~C. Myers.
\newblock A model for delimited release.
\newblock In {\em 2003 International Symposium on Software Security}, number
  3233 in Lecture Notes in Computer Science, pages 174--191. Springer-Verlag,
  2004.
\newblock URL: \url{http://www.cs.cornell.edu/andru/papers/isss03.pdf}.

\bibitem{schultz2013}
D.~A. Schultz and B.~Liskov.
\newblock {IFDB}: decentralized information flow control for databases.
\newblock In {\em EUROSYS}, 2013.

\bibitem{abrtchanTR}
I.~Sheff, T.~Magrino, J.~Liu, A.~C. Myers, and R.~van Renesse.
\newblock Safe serializable secure scheduling: Transactions and the trade-off
  between security and consistency.
\newblock Technical Report 1813--44581, Cornell University Computing and
  Information Science, Aug. 2016.
\newblock URL: \url{https://ecommons.cornell.edu/handle/1813/44581}.

\bibitem{dinosaurbook}
A.~Silberschatz, P.~Galvin, and G.~Gagne.
\newblock {\em Operating System Concepts}.
\newblock Windows XP update. Wiley, 2003.
\newblock URL: \url{https://books.google.com/books?id=9\_-oQgAACAAJ}.

\bibitem{Smith1996}
K.~Smith, B.~Blaustein, S.~Jajodia, and L.~Notargiacomo.
\newblock Correctness criteria for multilevel secure transactions.
\newblock {\em Knowledge and Data Engineering, IEEE Transactions on},
  8(1):32--45, Feb 1996.
\newblock \href {http://dx.doi.org/10.1109/69.485627}
  {\path{doi:10.1109/69.485627}}.

\bibitem{Suther86}
D.~Sutherland.
\newblock A model of information.
\newblock In {\em 9th National Security Conference}, pages 175--183,
  Gaithersburg, Md., 1986.

\bibitem{Wei2015}
X.~Wei, J.~Shi, Y.~Chen, R.~Chen, and H.~Chen.
\newblock Fast in-memory transaction processing using {RDMA} and {HTM}.
\newblock In {\em Proceedings of the 25th Symposium on Operating Systems
  Principles}, SOSP '15, pages 87--104, New York, NY, USA, 2015. ACM.
\newblock URL: \url{http://doi.acm.org/10.1145/2815400.2815419}, \href
  {http://dx.doi.org/10.1145/2815400.2815419}
  {\path{doi:10.1145/2815400.2815419}}.

\bibitem{Xie2015}
C.~Xie, C.~Su, C.~Littley, L.~Alvisi, M.~Kapritsos, and Y.~Wang.
\newblock High-performance {ACID} via modular concurrency control.
\newblock In {\em Proceedings of the 25th Symposium on Operating Systems
  Principles}, SOSP '15, pages 279--294, New York, NY, USA, 2015. ACM.
\newblock URL: \url{http://doi.acm.org/10.1145/2815400.2815430}, \href
  {http://dx.doi.org/10.1145/2815400.2815430}
  {\path{doi:10.1145/2815400.2815430}}.

\bibitem{zm03}
S.~Zdancewic and A.~C. Myers.
\newblock Observational determinism for concurrent program security.
\newblock In {\em 16\textsuperscript{th} {IEEE} Computer Security Foundations
  Workshop (CSFW)}, pages 29--43, June 2003.
\newblock URL: \url{http://www.cs.cornell.edu/andru/papers/csfw03.pdf}.

\bibitem{Zhang2015}
I.~Zhang, N.~K. Sharma, A.~Szekeres, A.~Krishnamurthy, and D.~R.~K. Ports.
\newblock Building consistent transactions with inconsistent replication.
\newblock In {\em Proceedings of the 25th Symposium on Operating Systems
  Principles}, SOSP '15, pages 263--278, New York, NY, USA, 2015. ACM.
\newblock URL: \url{http://doi.acm.org/10.1145/2815400.2815404}, \href
  {http://dx.doi.org/10.1145/2815400.2815404}
  {\path{doi:10.1145/2815400.2815404}}.

\end{thebibliography}
\fi

\end{document}